\documentclass[sigconf]{acmart}

\usepackage{styles/include}
\usepackage{styles/utilities}
\usepackage{styles/algorithm-style}
\usepackage{styles/problem-style}
\usepackage{styles/variables}
\AtBeginDocument{%
  \providecommand\BibTeX{{%
    \normalfont B\kern-0.5em{\scshape i\kern-0.25em b}\kern-0.8em\TeX}}}

\copyrightyear{2022}
\acmYear{2022}
\setcopyright{acmcopyright}
\acmConference[KDD '22] {Proceedings of the 28th ACM SIGKDD Conference on Knowledge Discovery and Data Mining}{August 14--18, 2022}{Washington, DC, USA.}
\acmBooktitle{Proceedings of the 28th ACM SIGKDD Conference on Knowledge Discovery and Data Mining (KDD '22), August 14--18, 2022, Washington, DC, USA}
\acmPrice{15.00}
\acmISBN{978-1-4503-9385-0/22/08}
\acmDOI{10.1145/3534678.3539371}

\begin{document}

\title{Minimizing Congestion for Balanced Dominators}

\author{Yosuke Mizutani}
\email{yos@cs.utah.edu}
\orcid{0000-0002-9847-4890}
\affiliation{%
  \institution{University of Utah}
  \streetaddress{P.O. Box 1212}
  \city{Salt Lake City}
  \state{Utah}
  \country{USA}
}

\author{Annie Staker}
\email{annie.staker@utah.edu}
\orcid{0000-0002-2106-4068}
\affiliation{%
  \institution{University of Utah}
  \streetaddress{P.O. Box 1212}
  \city{Salt Lake City}
  \state{Utah}
  \country{USA}
}

\author{Blair D. Sullivan}
\email{sullivan@cs.utah.edu}
\orcid{0000-0001-7720-6208}
\affiliation{%
  \institution{University of Utah}
  \streetaddress{P.O. Box 1212}
  \city{Salt Lake City}
  \state{Utah}
  \country{USA}
}


\begin{abstract}
  A primary challenge in metagenomics is reconstructing individual
microbial genomes from the mixture of short fragments created by sequencing.
Recent work leverages the sparsity of the assembly graph to find
$r$-dominating sets which enable rapid approximate queries
through a dominator-centric graph partition.
In this paper, we consider two problems related to reducing uncertainty and
improving scalability in this setting.

First, we observe that nodes with multiple closest dominators necessitate arbitrary
tie-breaking in the existing pipeline. As such, we propose finding \textit{sparse}
dominating sets which minimize this effect via a new \textit{congestion} parameter.
We prove minimizing congestion is NP-hard, and give an $\mathcal{O}(\sqrt{\Delta^r})$ approximation algorithm, where $\Delta$ is the max degree.

To improve scalability, the graph should be partitioned into uniformly sized pieces, subject to placing vertices with a closest dominator. This
leads to \textit{balanced neighborhood partitioning}: given an $r$-dominating set, find a partition into connected subgraphs with optimal uniformity so that each vertex is co-assigned with some closest dominator. Using variance of piece sizes to measure uniformity,
we show this problem is NP-hard iff $r$ is greater than $1$. We design and analyze
several algorithms, including a polynomial-time approach which is exact when $r=1$
(and heuristic otherwise).\looseness-1

We complement our theoretical results with computational experiments
 on a corpus of real-world networks showing sparse dominating sets lead to
 more balanced neighborhood partitionings. Further, on the metagenome \scalebox{.85}{\textsf{HuSB1}},
 our approach maintains high query containment and similarity
 while reducing piece size variance.

\end{abstract}




\maketitle

\section{Introduction}

Microbial communities play a critical role in many aspects of human health
 (e.g. gut microbiomes) and ecosystems (e.g. marine ecology),
 and understanding their composition and function has been increasingly important in
 biological and medical research.
Much of the work on these communities focuses on analyzing the genomic material (DNA and RNA)
 of the constituent microorganisms, a research area called metagenomics.
A primary challenge in the field is reconstructing individual genomes
 from the mixture of short fragments created by shotgun sequencing~\cite{quince2017shotgun}.

One practical approach that has gathered significant recent attention utilizes
 a metagenome assembly graph to guide analyses.
Commonly, this is done with a (compact) De Bruijn graph, or (c)DBG,
 where vertices correspond to DNA subsequences called $k$-mers and
 edges indicate potential compatibility in an assembly (almost complete overlap).
Since the graphs corresponding to real-world metagenomic datasets may have
 tens of millions of vertices, scalable methods for analysis are imperative.
 Recent work of Brown et al.~\cite{Brown2020}, implemented in
 \sgc \footnote{https://github.com/spacegraphcats/spacegraphcats},
 leveraged the sparsity of these graphs to enable efficient indexing and querying
 using partial information about suspected constituent microbes.
Their approach relies on finding an $r$-dominating set
 (using Dvorak's approximation algorithm for sparse graphs~\cite{dvorak2013constantfactor}),
 then partitioning the assembly graph into bounded-radius \emph{pieces} by assigning each vertex to
 one of its closest dominators.
The process is repeated on the piece graph to form a hierarchy of dominating sets
 which enables effective navigation and categorization of the data.
Initial experiments demonstrated the approach can improve the completeness of
 partial genomes for microbes present in a community
 and also reveal significant strain variation in real-world microbiomes \cite{Brown2020}.\looseness-1

Despite these promising results, several key challenges remain.
Here, we focus on those related to partitioning the metagenome assembly graph
 into pieces around the dominators.
In particular, the current dominating set algorithm in~\cite{Brown2020}
may choose dominators that are close together in the cDBG,
 leading to uncertainty in how a region should be ``carved up'' into pieces. Thus, the resulting piece graph may reflect the tie-breaking rules more than underlying ground-truth associations intrinsic to the cDBG/network.
Further, the sizes of the pieces may be imbalanced,
 counteracting any advantage gained by using the hierarchy to prune away irrelevant regions of the cDBG.

In this work, we tackle these challenges by (a) introducing a notion of
 \textit{sparse dominating sets} which rewards ``scattering'' dominators,
 with the aim of generating pieces which are biologically meaningful and
 inherently more stable,
 and (b) considering algorithms for \textit{balanced neighborhood partitioning}:
 assigning vertices to dominators which minimize variation in the resulting piece sizes.
 \looseness-1

While we defer formal definitions to Section~\ref{sec:sparse-domset},
 our approach to avoiding nodes with multiple closest dominators is
 based on minimizing \textit{congestion}, which measures the average number
of dominators appearing in an arbitrary vertex neighborhood.
We show that \PrbMCDS is NP-hard and present an $\mathcal{O}(\sqrt{\Delta^r})$-approximation algorithm running in $\bigo{\Delta^{2r} n \log n}$ time, where $\Delta$ denotes the maximum degree.
We compare this with the $\mathcal{O}(r\log{\Delta})$ standard approximation algorithm
 for finding a (smallest) \PrbRDom. We note that sparse dominating sets have no explicit
 size restriction, and discuss trade-offs between solution size and congestion.

Once we have an $r$-dominating set, the problem becomes one of partitioning
 the vertices into pieces so that (a) each vertex is assigned to a piece containing
 one of its closest dominators,
 and (b) the pieces are as equal in size as possible.
For the latter condition, we minimize the variance of the piece size distribution
 in \PrbBNP.
We show this is polynomial-time solvable when $r=1$ and NP-hard when $r \geq 2$,
 even when there are only two dominators.
Despite this, \AlgPrtBranch establishes the problem is fixed parameter tractable%
\footnote{%
  A problem is called fixed parameter tractable (FPT) with respect to parameter $k\in \N$
  if it can be solved in time $f(k)n^{O(1)}$ for some computable function $f$.
} (FPT) in graphs of bounded degree when parameterized by the number of vertices equidistant
 from multiple dominators.
Further, when $r=1$, we give an exact $\bigo{n^4}$ algorithm \AlgPrtLayer in general graphs
 using flow-based techniques; when $r \geq 2$, this yields a heuristic.
Finally, we compare with a linear-time greedy heuristic, \AlgPrtWeight.
These algorithms are described in Section~\ref{sec:nbr-prt}.

We implemented all of the above algorithms using C++ in an open-source repository
 and tested their performance on a large corpus of real-world networks,
 including a variety of metagenome graphs.
Experimental results demonstrate that the choice of dominating set can significantly impact
 the runtime and solution quality of balanced neighborhood partitioning algorithms,
 with sparse sets out-performing their smaller but more congested analogues.
Finally, we present preliminary results indicating that low-congestion dominating sets
 do not significantly degrade the fidelity of queries using partial genome bins on
  \scalebox{.85}{\textsf{HuSB1}}, a metagenome analyzed in~\cite{Brown2020},
 a critical requirement for their downstream adoption.

\section{Preliminaries}

\subsection{Notation \& Terminology}

Given a graph $G=(V,E)$, we write $n=\abs{G}=\abs{V}$ for the number of vertices and
$m=\norm{G}=\abs{E}$ for the number of edges. The \textit{distance} between two vertices
$x,y\in V$, denoted $d_G(x,y)$, equals the minimum number of edges in an $x,y$-path in
$G$; if no such path exists, we set $d(x,y) := \infty$. The distance between a vertex
$v \in V$ and vertex set $X \subseteq V$ is defined as
$d_G(v,X) := \min_{x\in X}d_G(v,x)$. We use $N(v)$ and $N[v]$ to denote the open and closed
neighborhoods of a vertex $v$, respectively.
We write $N^r[v]$ for the $r$-neighborhood: the set of all vertices at
distance at most $r$ from $v$. For a vertex set $X \subseteq V$, $N^r[X]$ denotes the union
of $N^r[x]$ for all $x \in X$. The $r$-th power of $G$ is defined as
$G^r := \left(V,\left\{uv \mid u,v \in V,\ d_G(u,v) \leq r\right\}\right)$.

We write $\deg_G(v)=\deg(v)$ for the degree of a vertex $v$,
 and $\delta$ and $\Delta$ for the minimum and maximum degree of $G$, respectively.
A graph is called \textit{regular} if $\delta=\Delta$.
We denote the induced subgraph of $G$ on a set $X \subseteq V$ by $G[X]$.
For an edge $e \in E$, we write $G/e$ for the graph formed from $G$ by contracting the edge
$e$; for a vertex set $X \subseteq V$, $G/X$ denotes the graph obtained by contracting all edges in
$G[X]$.\looseness-1

\subsection{Dominating Sets \& Related Problems}

For a graph $G=(V,E)$, a vertex set $D \subseteq V$ is called a \textit{dominating set}
if $N[D]=V$.
An \textit{$r$-dominating set} generalizes this notion and is defined as a set $D \subseteq V$
such that $N^r[D]=V$.
The problems asking for such sets of minimum cardinality are
called \PrbMDS (\PrbMDSShort) and \PrbMRDS, respectively, and are NP-complete,
even for regular graphs of degree 4 \cite{garey1979computers}.
An \textit{$r$-perfect code} is an $r$-dominating set $D$ such that every vertex $v \in V$
satisfies $\abs{N^r[v] \cap D} = 1$.
It is NP-complete to determine whether a graph has a perfect code \cite{KRATOCHVIL1994191,marilynn1997}.

The problem \PrbMRDS is further generalized by \PrbMWDS (\PrbMWDSShort), where each vertex has non-negative weight,
and one seeks an $r$-dominating set of minimum total weight.
Although \PrbMDSShort and \PrbMWDSShort cannot be approximated within a factor $c \log \abs{V}$
for some constant $c>0$ in polynomial time (unless P=NP) \cite{raz1997sub, alon2006algorithmic},
it is known that \PrbMWDSShort has a polynomial-time approximation scheme (PTAS)
in growth-bounded graphs with bounded degree constraint \cite{wang2012ptas}.

\section{Sparse Dominating Sets}
\label{sec:sparse-domset}

In this section, we formulate the problem of finding \textit{sparse dominating sets},
 establish its hardness, and describe several heuristics along with
 an integer linear programming (ILP) formulation.

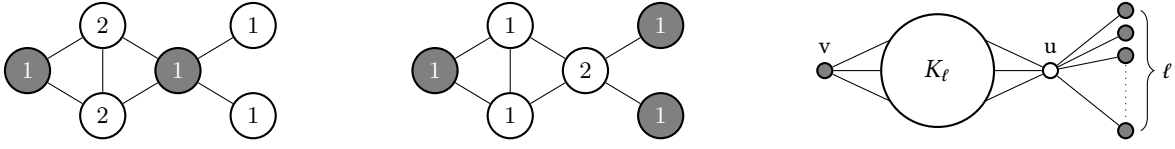
\begin{figure*}[ht]
    \centering

    \pgfdeclarelayer{bg}
    \pgfsetlayers{bg, main}

    \tikzstyle{bigblacknode} = [circle, fill=gray, text=white, draw, thick, scale=1, minimum size=0.6cm, inner sep=1.5pt]
    \tikzstyle{bigwhitenode} = [circle, fill=white, text=black, draw, thick, scale=1, minimum size=0.6cm, inner sep=1.5pt]

    \tikzstyle{blacknode} = [circle, fill=gray, draw, thick, scale=1, minimum size=0.2cm, inner sep=1.5pt]
    \tikzstyle{whitenode} = [circle, fill=white, draw, thick, scale=1, minimum size=0.2cm, inner sep=1.5pt]

    \tikzstyle{hugewhitenode} = [circle, fill=white, text=black, draw, thick, scale=1, minimum size=1.5cm, inner sep=1.5pt, font=\large]

    \begin{minipage}[m]{.30\linewidth}
        \vspace{0pt}
        \centering
        \begin{tikzpicture}
            \node[bigblacknode] (x1) at (0, 0) {\textbf{$1$}};
            \node[bigwhitenode] (x2) at (1, 0.6) {$2$};
            \node[bigwhitenode] (x3) at (1, -0.6) {$2$};
            \node[bigblacknode] (x4) at (2, 0) {\textbf{$1$}};
            \node[bigwhitenode] (x5) at (3, 0.6) {$1$};
            \node[bigwhitenode] (x6) at (3, -0.6) {$1$};

            \draw (x1) -- (x2);
            \draw (x1) -- (x3);
            \draw (x2) -- (x3);
            \draw (x2) -- (x4);
            \draw (x3) -- (x4);
            \draw (x4) -- (x5);
            \draw (x4) -- (x6);
        \end{tikzpicture}
    \end{minipage}
    \begin{minipage}[m]{.30\linewidth}
        \vspace{0pt}
        \centering
        \begin{tikzpicture}
            \node[bigblacknode] (x1) at (0, 0) {\textbf{$1$}};
            \node[bigwhitenode] (x2) at (1, 0.6) {$1$};
            \node[bigwhitenode] (x3) at (1, -0.6) {$1$};
            \node[bigwhitenode] (x4) at (2, 0) {$2$};
            \node[bigblacknode] (x5) at (3, 0.6) {\textbf{$1$}};
            \node[bigblacknode] (x6) at (3, -0.6) {\textbf{$1$}};

            \draw (x1) -- (x2);
            \draw (x1) -- (x3);
            \draw (x2) -- (x3);
            \draw (x2) -- (x4);
            \draw (x3) -- (x4);
            \draw (x4) -- (x5);
            \draw (x4) -- (x6);
        \end{tikzpicture}
    \end{minipage}
    \begin{minipage}[m]{.36\linewidth}
        \vspace{0pt}
        \centering
        \begin{tikzpicture}
            \node[blacknode] (v) at (0, 0) [label=v] {};
            \node[whitenode] (u) at (3, 0) [label=u] {};
            \node[hugewhitenode] (k) at (1.5, 0) {$K_\ell$};
            \node[blacknode] (y0) at (4, 0.8) {};
            \node[blacknode] (y1) at (4, 0.5) {};
            \node[blacknode] (y2) at (4, 0.2) {};
            \node[blacknode] (y3) at (4, -0.8) {};

            \begin{pgfonlayer}{bg}
                \draw (u) -- (y0);
                \draw (u) -- (y1);
                \draw (u) -- (y2);
                \draw (u) -- (y3);
                \draw (v) -- (k) -- (u);
                \draw (v) -- (1.5,0.7) -- (u);
                \draw (v) -- (1.5,-0.7) -- (u);
                \draw[dotted] (4,-0.7) -- (4,0.1);

                \draw [decorate,decoration={brace,amplitude=5pt,mirror}]
  (4.2,-0.8) -- (4.2, 0.8) node[midway,xshift=10pt]{$\ell$};
            \end{pgfonlayer}
        \end{tikzpicture}
    \end{minipage}
    \caption{Examples of the difference between solutions to \PrbMDS and \PrbMCDSShort. Dominators are shaded in gray, and vertices are labelled with their congestion. A minimum dominating set (left, size $2$) has average congestion $8/6$, whereas a size $3$ dominating set (center) achieves average congestion $7/6$. At right, the dominating set $\{v,u\}$ has minimum size, but the set shaded in gray has lower congestion. As $\ell$ grows, the size of this dominating set can be arbitrarily large.}
    \label{fig:sparse-domset}
\end{figure*}

To define a \textit{sparse} dominating set, we first introduce the notion of \textit{congestion}
 which measures how frequently a given set of vertices overlaps with the $r$-neighborhoods
 in a graph.

\begin{definition}
    Given a graph $G=(V,E)$, vertex set $S \subseteq V$, and radius $r \in \N$,
    the \textbf{$r$-congestion} of $S$ at a vertex $v \in V$, denoted $\rcong(S, v)$,
    is $\abs{N^r[v] \cap S}$.
    The \textbf{average $r$-congestion} of $S$ in $G$ is then
    $\ravgcong(S) = \frac{1}{|V|}\sum_{v \in V}{\rcong(S, v)}$.
\end{definition}

We observe that the average congestion of a given set $S$ be computed directly
 from the neighborhood sizes of vertices in $S$.

\begin{lemma}\label{lem:prop-avgcong}
    Average congestion can be computed as
    $\ravgcong(S) = \frac{1}{\abs{V}}\sum_{u \in S}{\abs{N^r[u]}}$.
\end{lemma}

We say an $r$-dominating set is \textit{sparse} when it achieves low average $r$-congestion,
naturally leading to the following problem.

\begin{ProblemBox}{\small \PrbMCDS (\PrbMCDSShort) \normalsize}
    \Input & A graph $G=(V,E)$ and radius $r \in \N$.\\
    \Prob  & Find an $r$-dominating set $D \subseteq V$ such that
    $\ravgcong(D)$ is minimized.\\
\end{ProblemBox}

We remark that this is distinct from the class of problems studied
 in~\cite{JAFFKE2019216, einarson2020general} which put uniform local constraints
 on each vertex (e.g. that they are dominated at least $\lambda$ and at most $\mu$ times).

We write $\mac^r(G)$ ($\mac(G)$ when $r=1$) for the minimum average congestion attainable
 by any $r$-dominating set on $G$.
By Lemma~\ref{lem:prop-avgcong}, $\abs{V}\cdot\ravgcong(D)$
 equals the weighted sum over $D$ of $w(v)=\abs{N^r(v)}$.
Thus, \PrbMCDSShort is a specialization of \PrbMWDSShort.
Furthermore, like other dominating set problems~\cite{slater1976}, \PrbMCDSShort$(G,r)$ is
 equivalent to \PrbMCDSShort$(G^r, 1)$, where $G^r$ denotes
 the $r$\textsuperscript{th} power of $G$.
We now establish that minimizing average congestion is NP-hard.

\begin{theorem}
    \PrbMCDSShort is NP-hard.
    \label{thm:hardness-mac}
\end{theorem}

\begin{proof}
    We show that \PrbMCDSShort with $r=1$ is equivalent to \PrbMDSShort when
    $G$ is regular.
    By Lemma~\ref{lem:prop-avgcong}, if $G$ is $d$-regular, then
    $\ravgcong(S) = \frac{1}{\abs{V}}\sum_{u \in S} \abs{N_{G}[u]} =
    \left(1+d\right)\abs{S}/\abs{V}$.
    Thus, any minimum congestion dominating set must also be minimum in size.
    The result follows directly, since $\PrbMDSShort$ is NP-hard in regular graphs \cite{garey1979computers}.
\end{proof}

Additionally, we observe that determining the value of $\mac_r(G)$ is hard
via its relationship to perfect codes. Specifically, $G$ admits an $r$-perfect code if and only if $\mac_r(G)=1$, but determining the existence of a perfect code is NP-hard \cite{KRATOCHVIL1994191,marilynn1997}.


A minimum congestion dominating set thus represents the ``distance'' to a perfect code,
leading to a natural graph editing problem whose optimal solution is bounded by a linear function of $\mac(G)$.

\begin{ProblemBox}{\PrbPCE}
    \Input & A graph $G=(V,E)$ and integer $k$.\\
    \Prob  & Is there an edge set $S \subseteq E$ of size at most $k$ such that
            $G'=(V,E\setminus S)$ admits a perfect code?
\end{ProblemBox}

\begin{theorem}\label{thm:prop-pce}
    \raggedright
    Let $\textnormal{PCE}(G)$ be the minimum $k$ so $(G,k)$ is a yes-instance of \PrbPCE.
    Then $\textnormal{PCE}(G) \leq (\mac (G) - 1)n \leq 2\cdot \textnormal{PCE}(G)$.
\end{theorem}

\begin{proof}
    Given a dominating set $D \subseteq V$ that attains $\mac(G)$, let $v \in V$ be
    a vertex such that $\cong(D,v) > 1$.
    Then, removing an edge $uv$ for any $u \in D \setminus \{v\}$ will decrease
    $(\mac(G)\cdot n)$ by at least $1$, so $\textnormal{PCE}(G)\leq (\mac(G)-1)n$.
    Given a perfect code $D' \subseteq V$ of $G'$, any edge addition can increase
    $(\mac(G)\cdot n)$ by at most $2$, so $(\mac(G)-1)n \leq 2\cdot\textnormal{PCE}(G)$.
\end{proof}

\subsection{Properties of Minimum Congestion Dominating Sets}

In general, a minimum congestion dominating set will not also be a minimum dominating set,
 and in Figure~\ref{fig:sparse-domset} we give a construction proving their sizes can diverge
 arbitrarily.
Further, by definition, we have $\mac(G) \geq 1$ for any graph;
 we give a degree-based upper bound below.

\begin{theorem}
    $\mac(G)\leq (\overline{d}+1)/2$ for every graph $G$,
     where $\overline{d}$ is the average degree of $G$.
    \label{thm:prop-cong-bound}
\end{theorem}

\begin{proof}
    Let $D \subseteq V$ be a minimal dominating set
    (i.e. every proper subset $D' \subset D$ is not dominating).
     Then, $\overline{D}:=V\setminus D$ is also a dominating set. Now,
    $\avgcong(D) + \avgcong(\overline{D}) =
    \frac{1}{n}\sum_{v \in D}\abs{N[v]} + \frac{1}{n}\sum_{v \in \overline{D}}\abs{N[v]}$.
    Rewriting the right-hand side as $\frac{1}{n}\sum_{v \in V}\abs{N[v]} = \overline{d} + 1$,
    we have $\mac(G)\leq \min\{\avgcong(D), \avgcong(\overline{D})\} \leq (\overline{d} + 1) / 2$.
\end{proof}



\subsection{Algorithms}

We now describe several greedy algorithms along with an ILP formulation for \PrbMCDSShort.

\subsubsection{Greedy Algorithms}

We first recall the standard greedy algorithm for \PrbMDS
 which we call \AlgDegree.
At each step, the algorithm chooses a vertex $v \in V$ such that
 the number of undominated vertices in $N^r[v]$ is maximized.
To instead target minimizing congestion, we prioritize based on
 the \textit{ratio} of the undominated vertices.
Specifically, \AlgRatio chooses a vertex $v \in V$ such that
 $\frac{\abs{N^r[v] \setminus N^r[D]}}{\abs{N^r[v]}}$ is maximized
 (equivalently, $\frac{\abs{N^r[v] \cap N^r[D]}}{\abs{N^r[v]}}$ is minimized)
 at each step given
 a partial dominating set $D \subseteq V$.
In both algorithms, ties are broken arbitrarily.
While \AlgDegree is an $\bigo{r\log\Delta}$-approximation for \PrbRDom,
 it is not for \PrbMCDSShort (Theorem~\ref{thm:dom-degree-unbounded}).

\begin{theorem}\label{thm:dom-degree-unbounded}
    \AlgDegree does not approximate \PrbMCDSShort.
\end{theorem}

\begin{proof}
    Let $r=1$. Consider a graph $G=(V,E)$ that has a biclique $(A,B)$ with
    $\abs{A}=\abs{B}=k$ and one attached leaf $x_v$ for each $v \in A \cup B$. Thus,
    $n=\abs{V}=4k$. In the first two iterations, \AlgDegree should choose vertices $u \in A$
    and $v \in B$. But then, all $4(k-1)$ vertices in $V \setminus \{u, v, x_u, x_v\}$
    equally have one undominated vertex in their neighborhoods. Thus, \AlgDegree may choose
    $D := A \cup B$, which gives $\avgcong(D)=n/8 + 1$.
    However, $\overline{D} := V \setminus D$ is a perfect code, so $\mac(G)=1$.\looseness-1
\end{proof}

In contrast, \AlgRatio can produce sets which are arbitrarily bigger
 than the minimum dominating set.
 (Figure~\ref{fig:sparse-domset} (right) is an example of this),
 yet we prove it is an approximation for \PrbMCDSShort\footnote{An abridged proof of Theorem~\ref{thm:approx-ratio-greedy} is in Appendix~\ref{sec:theory}.}:

\begin{theorem}
    \AlgRatio is an $\bigo{\sqrt{\Delta^r}}$-approxima\-tion algorithm for \PrbMCDSShort.
    \label{thm:approx-ratio-greedy}
\end{theorem}

To evaluate smarter tie-breaking strategies, we define \linebreak
 \AlgDegreePlus to be \AlgDegree with ties broken using \AlgRatio's criteria
 (the ratio of undominated vertices),
 and \AlgRatioPlus analogously using \AlgDegree's criteria (the number of undominated vertices).
Further ties are randomly broken.
All of these algorithms run in time $\bigo{\Delta^{2r}n \log n}$ (details in Appendix~\ref{sec:alg-details}).

\subsubsection{Integer Programming}

We observe that one may obtain optimal solutions to \PrbRDom and \PrbMCDSShort
using an ILP solver, allowing empirical evaluation of approximation ratios
in Section~\ref{sec:exp}.
We use the following ILP formulation for \PrbRDom:
\begin{LabelBox}{\AlgMDS}
    \begin{talign*}
        \text{Let } &x_v \in \{0,1\}\text{ be variables} &\text{for all }v \in V.\\
        \text{Minimize } &\sum_{v \in V} x_v\\
        \text{Subject to } &\sum_{w \in N^r[v]}x_w \geq 1 &\text{for all }v \in V.
    \end{talign*}
\end{LabelBox}

Similarly, we formulate \PrbMCDSShort as follows:
\begin{LabelBox}{\AlgMAC}
    \begin{talign*}
        \text{Let } &x_v \in \{0,1\}\text{ be variables} &\text{for all }v \in V.\\
        \text{Minimize } &\sum_{v \in V} \abs{N^r[v]}x_v\\
        \text{Subject to } &\sum_{w \in N^r[v]}x_w \geq 1 &\text{for all }v \in V.
    \end{talign*}
\end{LabelBox}

This is based on the fact that \PrbMCDSShort is a specialization of \PrbMWDS.
In both cases, the solution set is provided by $\{v \in V: x_v=1\}$.

\section{Neighborhood Partitioning}
\label{sec:nbr-prt}

\begin{figure*}[ht]
    \centering

    \pgfdeclarelayer{bg}
    \pgfsetlayers{bg, main}

    \tikzstyle{bigblacknode} = [circle, fill=gray, text=white, draw, thick, scale=1, minimum size=0.6cm, inner sep=1.5pt]
    \tikzstyle{bigwhitenode} = [circle, fill=white, text=black, draw, thick, scale=1, minimum size=0.6cm, inner sep=1.5pt]

    \tikzstyle{blacknode} = [circle, fill=gray, draw, thick, scale=1, minimum size=0.2cm, inner sep=1.5pt]
    \tikzstyle{whitenode} = [circle, fill=white, draw, thick, scale=1, minimum size=0.2cm, inner sep=1.5pt]

    \definecolor{color1}{RGB}{99,110,250}
    \definecolor{color4}{RGB}{177,183,253}
    \definecolor{color2}{RGB}{239,85,59}
    \definecolor{color5}{RGB}{247,170,157}
    \definecolor{color3}{RGB}{0,204,150}
    \definecolor{color6}{RGB}{0,230,203}

    \begin{minipage}{0.48\textwidth}
        \begin{tikzpicture}
            \begin{scope}[every node/.style=bigwhitenode]
                \node[fill=color4] (c1) at (1.5, 2) {$C_1$};
                \node[fill=color6] (c2) at (3.5, 2) {$C_2$};
                \node[fill=color5] (c3) at (5.5, 2) {$C_3$};

                \node[fill=color4] (x1) at (1, 0) {$x_1$};
                \node[fill=color5] (x2) at (2, 0) {$x_2$};
                \node[fill=color5] (x3) at (3, 0) {$x_3$};
                \node[fill=color5] (x4) at (4, 0) {$x_4$};
                \node[fill=color4] (x5) at (5, 0) {$x_5$};
                \node[fill=color4] (x6) at (6, 0) {$x_6$};
            \end{scope}

            \node[bigblacknode, fill=color1] (u1) at (1.5, 3) {\textbf{$u_1$}};
            \node[bigblacknode, fill=color2] (u2) at (3.5, 3) {\textbf{$u_2$}};
            \node[bigblacknode, fill=color3] (u3) at (5.5, 3) {\textbf{$u_3$}};

            \node[whitenode, fill=color6] (z1) at (5.2, 4) {};
            \node[whitenode, fill=color6] (z2) at (5.5, 4) {};
            \node[whitenode, fill=color6] (z3) at (5.8, 4) {};

            \begin{pgfonlayer}{bg}
                \foreach \i in {1,...,3} {
                    \foreach \j in {1,...,3} { \draw (u\i) -- (c\j); }
                    \draw (u3) -- (z\i);
                }
                \draw (c1) -- (x1);
                \draw (c1) -- (x5);
                \draw (c1) -- (x6);
                \draw (c2) -- (x1);
                \draw (c2) -- (x2);
                \draw (c2) -- (x3);
                \draw (c3) -- (x2);
                \draw (c3) -- (x3);
                \draw (c3) -- (x4);
            \end{pgfonlayer}

            \draw[dashed, rounded corners=1ex] (0.5, -0.5) rectangle (6.5, 0.5);
            \draw[dashed, rounded corners=1ex] (1, 1.6) rectangle (6, 2.4);
            \draw[dashed, rounded corners=1ex] (1, 2.6) rectangle (6, 3.4);

            \node[align=right, anchor=east] at (8.2, 3) {landmarks};
            \node[align=left, anchor=west] at (6, 3.2) {$L$};
            \node[align=left, anchor=west] at (6, 2.2) {$\mathcal{C}$};
            \node[align=left, anchor=west] at (6.5, 0.2) {$X$};
            \node[align=right, anchor=east] at (8.2, 2) {$3$-sets};
            \node[align=right, anchor=east] at (8.2, 0) {elements};
        \end{tikzpicture}

        \captionof{figure}{Reduction from \PrbXTC to \PrbBNPShort,
        showing an example: $C_1=\{x_1,x_5,x_6\}, C_2=\{x_1,x_2,x_3\}, C_3=\{x_2,x_3,x_4\}$ with
        solution $\{C_1,C_3\}$.
        Landmarks are a $2$-dominating set.}
        \label{fig:nbr-hardness-1}
    \end{minipage}
    \hfill
    \begin{minipage}{0.48\textwidth}
        \begin{tikzpicture}
            \node[bigblacknode, fill=color1] (u1) at (2.2, 3) {\textbf{$u_1$}};
            \node[bigblacknode, fill=color2] (u2) at (5.4, 3) {\textbf{$u_2$}};

            \begin{scope}[every node/.style=bigwhitenode]
                \node[fill=color4] (x1t) at (1, 2) {$x_1$};
                \node[fill=color5] (x1f) at (1.8, 2) {$\overline{x}_1$};
                \node[fill=color5] (x2t) at (2.6, 2) {$x_2$};
                \node[fill=color4] (x2f) at (3.4, 2) {$\overline{x}_2$};
                \node[fill=color4] (x3t) at (4.2, 2) {$x_3$};
                \node[fill=color5] (x3f) at (5.0, 2) {$\overline{x}_3$};
                \node[fill=color5] (x4t) at (5.8, 2) {$x_4$};
                \node[fill=color4] (x4f) at (6.6, 2) {$\overline{x}_4$};

                \node[fill=color4] (y1) at (5.5, 0.5) {$y_1$};
                \node[fill=color4] (y2) at (6.4, 0.5) {$y_2$};
                \node[fill=color4] (y3) at (7.3, 0.5) {$y_3$};
                \node[fill=color4] (y4) at (8.2, 0.5) {$y_4$};

                \node[fill=color4] (c1) at (1, 0.5) {$\phi_1$};
                \node[fill=color4] (c2) at (1.9, 0.5) {$\phi_2$};
                \node[fill=color4] (c3) at (2.8, 0.5) {$\phi_3$};
            \end{scope}

            \begin{scope}[every node/.style=whitenode]
                \node[fill=color4] (zc11) at ([shift={(-0.3, -0.8)}] c1) {};
                \node[fill=color4] (zc12) at ([shift={(   0, -0.8)}] c1) {};
                \node[fill=color4] (zc13) at ([shift={( 0.3, -0.8)}] c1) {};
                \node[fill=color4] (zc21) at ([shift={(-0.3, -0.8)}] c2) {};
                \node[fill=color4] (zc22) at ([shift={(   0, -0.8)}] c2) {};
                \node[fill=color4] (zc23) at ([shift={( 0.3, -0.8)}] c2) {};
                \node[fill=color4] (zc31) at ([shift={(-0.3, -0.8)}] c3) {};
                \node[fill=color4] (zc32) at ([shift={(   0, -0.8)}] c3) {};
                \node[fill=color4] (zc33) at ([shift={( 0.3, -0.8)}] c3) {};

                \node[fill=color4] (zy11) at ([shift={(-0.3, -0.8)}] y1) {};
                \node[fill=color4] (zy12) at ([shift={(   0, -0.8)}] y1) {};
                \node[fill=color4] (zy13) at ([shift={( 0.3, -0.8)}] y1) {};
                \node[fill=color4] (zy21) at ([shift={(-0.3, -0.8)}] y2) {};
                \node[fill=color4] (zy22) at ([shift={(   0, -0.8)}] y2) {};
                \node[fill=color4] (zy23) at ([shift={( 0.3, -0.8)}] y2) {};
                \node[fill=color4] (zy31) at ([shift={(-0.3, -0.8)}] y3) {};
                \node[fill=color4] (zy32) at ([shift={(   0, -0.8)}] y3) {};
                \node[fill=color4] (zy33) at ([shift={( 0.3, -0.8)}] y3) {};
                \node[fill=color4] (zy41) at ([shift={(-0.3, -0.8)}] y4) {};
                \node[fill=color4] (zy42) at ([shift={(   0, -0.8)}] y4) {};
                \node[fill=color4] (zy43) at ([shift={( 0.3, -0.8)}] y4) {};

                \node[fill=color5] (zu1) at (4.3,4) {};
                \node[fill=color5] (zu2) at (4.6,4) {};
                \node[fill=color5] (zu3) at (4.9,4) {};
                \node[fill=color5] (zux) at (6.5,4) {};
            \end{scope}

            \begin{pgfonlayer}{bg}
            \foreach \i in {1,...,4} {
                \draw (u1) -- (x\i t.north);
                \draw (u1) -- (x\i f.north);
                \draw (u2) -- (x\i t.north);
                \draw (u2) -- (x\i f.north);
                \draw (y\i.north) -- (x\i t.south);
                \draw (y\i.north) -- (x\i f.south);
            }
            \draw (c1.north) -- (x1t.south);
            \draw (c1.north) -- (x3f.south);
            \draw (c1.north) -- (x4f.south);
            \draw (c2.north) -- (x1f.south);
            \draw (c2.north) -- (x2t.south);
            \draw (c2.north) -- (x4f.south);
            \draw (c3.north) -- (x2t.south);
            \draw (c3.north) -- (x3t.south);
            \draw (c3.north) -- (x4t.south);

            \foreach \i in {1,...,3} {
                \foreach \j in {1,...,3} { \draw (c\j) -- (zc\j\i);}
                \foreach \j in {1,...,4} { \draw (y\j) -- (zy\j\i);}
            }

            \draw[densely dotted] (5.2,4) -- (6.2,4);
            \draw (u2) -- (zu1);
            \draw (u2) -- (zu2);
            \draw (u2) -- (zu3);
            \draw (u2) -- (zux);
            \end{pgfonlayer}

            \draw[dashed, rounded corners=1ex] (0.6, 0.1) rectangle (3.2, 0.9);
            \draw[dashed, rounded corners=1ex] (5.1, 0.1) rectangle (8.6, 0.9);
            \draw[dashed, rounded corners=1ex] (0.6, 1.6) rectangle (7, 2.4);
            \draw[dashed, rounded corners=1ex] (1.6, 2.6) rectangle (6, 3.4);

            \node[align=right, anchor=east] at (9,2) {literals};
            \node[align=left, anchor=west] at (7,2.2) {$X$};
            \node[align=right, anchor=east] at (9,3) {landmarks};
            \node[align=left, anchor=west] at (6,3.2) {$L$};
            \node[align=left, anchor=west] at (3.2,0.4) {clauses};
            \node[align=left, anchor=west] at (3.2,0.7) {$\Phi$};
            \node[align=right, anchor=east] at (9,1.3) {variables};
            \node[align=left, anchor=west] at (8.6,0.7) {$Y$};
            \node[align=left, anchor=west] at (3.2,-0.2) {
                \footnotesize\shortstack[l]{$n$-$1$=$3$ leaves\vspace*{-0.2em}\\for each $\phi,y$}
            };
            \node[align=right, anchor=east] at (9.1,4) {
                \footnotesize\shortstack[r]{$n(n+m)$\vspace*{-0.2em}\\$=28$ leaves}
            };
        \end{tikzpicture}
        \captionof{figure}{Reduction from \PrbTSATShort to \PrbBNPShort,
        showing an example with $n=4,m=3$: $\phi_1=x_1 \vee \overline{x}_3 \vee \overline{x}_4,
        \phi_2=\overline{x}_1 \vee x_2 \vee \overline{x}_4,
        \phi_3=x_2 \vee x_3 \vee x_4
        $ with solution $x_1=1,x_2=0,x_3=1,x_4=0$.
        The two landmarks ($u_1,u_2$) form a 3-dominating set.}
        \label{fig:nbr-hardness-2}
    \end{minipage}
\end{figure*}

We now turn to the second problem arising in our metagenomics application:
partitioning the vertex set into pieces around a set of $r$-dominators
as evenly as possible.
We first formalize the notion of a \textit{neighborhood partitioning} and
use variance to define its \textit{balance}.
After establishing that the resulting problem is NP-hard
(and remains so under very restrictive conditions),
we show that for radius $1$, a flow-based approach gives a polynomial-time solution.
Finally, we describe several algorithmic approaches for obtaining both exact and heuristic
solutions.

We begin by considering the more general setting where we are given $L$,
a set of \textit{landmarks}\footnote{
    We define landmarks to be any set of vertices $L$ so that every
     $v \in V$ is reachable from at least one $u \in L$.
},
and ask for a partition of $V$ into $|L|$ disjoint sets so that
(a) each piece contains exactly one landmark,
(b) every vertex is assigned to a piece with one of its closest landmarks from $L$ and
(c) for every piece $A$, the induced subgraph on $A$ preserves the distance between
 the landmark and other vertices in $A$.\looseness-1

\begin{definition}
  Given a graph $G=(V,E)$ and landmarks $L = \{u_1,\ldots,u_{\ell}\} \subseteq V$, we say
  $\mathcal{A}=\{A_1,\ldots,A_{\ell}\}$ is a \textbf{neighborhood partitioning} of $G$
  with respect to $L$ if and only if $\mathcal{A}$ is a set partition of $V$ and
  $d_{G[A_i]}(v,u_i) = d_G(v,L)$ for every $1 \leq i \leq \ell$, $v \in A_i$.
\end{definition}

We note that if $L$ is an $r$-dominating set of $G$, the resulting pieces will necessarily
have radius at most $r$, making $G/\mathcal{A}$ an $r$-shallow minor\footnote{
    A graph minor formed by contracting disjoint subgraphs of radius at most $r$.
}.
This is essential for \sgc to maintain efficiency guarantees when
computing the dominating sets in their hierarchy, since it ensures graphs remain
within the assumed sparse class (bounded expansion) \cite{Brown2020}.

We now define the problem of finding a set of pieces whose size distribution is
as even as possible.

\begin{ProblemBox}{\small \PrbBNP (\PrbBNPShort) \normalsize}
    \Input & A graph $G=(V,E)$ and landmarks $L = \{u_1,\ldots,u_{\ell}\} \subseteq V$\\
    \Prob  & Find a neighborhood partitioning $\mathcal{A}=\{A_1,\ldots,A_{\ell}\}$ of $G$
    on $L$ such that the piece-size population variance
    $\text{Var} (\{ \abs{A}: A \in \mathcal{A}\})$ is minimized.
\end{ProblemBox}

A key observation is that since $V$ and $L$ are given to us,
the average size of $\mathcal{A}$ is $\abs{V}/\abs{L}$, which is fixed.
This results in the following equivalence:

\begin{theorem}\label{thm:min-var-square-relation}
    A neighborhood partitioning $\mathcal{A}$ gives the minimum piece-size variance
    if and only if the square sum of the piece-sizes is minimized.
\end{theorem}

We now establish the hardness of \PrbBNPShort.

\begin{theorem}\label{thm:hardness-bnp}
    \PrbBNPShort is NP-hard.
    It remains hard even if \linebreak $\max_{v\in V}d(v,L) = 2$ or if $|L|=2$.
\end{theorem}

We break the proof of Theorem~\ref{thm:hardness-bnp} into two lemmas,
    first showing the case when the landmarks are a 2-dominating-set by
    reduction from \PrbXTC.

    \begin{ProblemBox}{\PrbXTC (\PrbXTCShort)}
        Input: &Set $X$ with $|X|=3q$ and a collection $\mathcal{C}$ of 3-element subsets of
        $X$.\\
        Problem: &Is there a collection $\mathcal{C'} \subseteq \mathcal{C}$ such that
        every element of $X$ occurs in exactly one member of $\mathcal{C'}$?
    \end{ProblemBox}

\begin{lemma}
    \PrbBNPShort is NP-hard when the landmarks are distance at most two from all vertices.
\end{lemma}

\begin{proof}
    Consider an instance of \PrbXTCShort with $n=|\mathcal{C}|$.
    If $n < q$, then output ``no''.
    Assuming $n \geq q$, we construct a graph $G=(V,E)$ as follows.
    The set of vertices $V$ includes $n$ landmarks
     $L=\{u_i \mid 1 \leq i \leq q\}$, $\mathcal{C}$, $X$, and attached leaves;
     $3$ leaves are attached to each of the $n-q$ vertices in $L$.
    The edges are constructed as follows:
     $L$ and $\mathcal{C}$ form a biclique;
     for each $C \in \mathcal{C}$, $C$ is connected to all elements of $X$ it contains.
    This construction can be done in time $\bigo{n^2}$.
    See Figure~\ref{fig:nbr-hardness-1} for an example.

    We return ``yes'' if and only if \PrbBNPShort returns a partition with pieces of equal size
    (by construction, this must be $5$). Let $L' \subseteq L$ be the landmarks that do not
    have attached leaves. If the given \PrbXTCShort instance admits an exact cover
    $\mathcal{C'}$, then for each $C \in \mathcal{C'}$, we assign one landmark in $L'$ to $C$
    and all elements of $C$. Then, we assign each of $L \setminus L'$ to $n - q$ unused sets
    in $\mathcal{C}$, resulting in pieces of equal sizes.

    Conversely, suppose \PrbBNPShort returns an equally-sized partition. Because each piece
    must have $5$ vertices, $n-q$ sets in $\mathcal{C}$ are assigned to $L \subseteq L'$.
    Also, since $\abs{L'}=q$, each landmark in $L'$ must have exactly one vertex of
    $\mathcal{C}$ in its piece. Then, each piece contains exactly $3$ elements in $X$, and this
    is an exact cover by $3$-sets.\looseness-1
\end{proof}

\begin{lemma}
    \PrbBNPShort is NP-hard when there are two landmarks.
\end{lemma}

\begin{proof}
    We reduce from \PrbSAT (\PrbSATShort).
    Let $x_1,\ldots,x_n$ be the variables and $\phi_1,\ldots,\phi_m$ be the clauses appearing in
    a \PrbSATShort instance. We construct a graph $G=(V,E)$ as follows.
    The set of vertices $V$ includes two landmarks $L=\{u_1, u_2\}$,
     $2n$ vertices $X=\{x_1,\overline{x}_1,\ldots,x_n,\overline{x}_n\}$
     representing \PrbSATShort literals,
     $n$ variables $Y=\{y_1,\ldots,y_n\}$,
     $m$ clauses $\Phi=\{\phi_1,\ldots,\phi_m\}$,
     and attached leaves.
    A total of $n-1$ leaves are attached to each of $Y$ and $\Phi$,
     and $n(n+m)$ leaves are attached to $u_2$.
    The set of edges $E$ is constructed as follows:
     $L$ and $X$ form a biclique;
     each $y_i \in Y$, is connected to $x_i$ and $\overline{x}_i$;
     each $\phi_j \in \Phi$ is connected to all literals in clause $\phi_j$.
    This construction can be done in time $\bigo{n(n+m)}$.
    See Figure~\ref{fig:nbr-hardness-2} for an example.

    Let $\{A, B\}$ be a partition such that $u_1 \in A$ and $u_2 \in B$.
    Note that in any neighborhood partitioning, each piece must be connected,
    so $B$ must contain all $n(n+m)$ of the leaves attached to it. Likewise, if a piece includes $Y_i$ or $\Phi_j$, it must also include their attached leaves.

    We return ``yes'' if \PrbBNPShort returns a partition of equal sizes.
    If the given \PrbSATShort instance is satisfiable, then we include $Y$, $\Phi$, and all true
    literals of $X$ in $A$. $B$ includes all other vertices, and $\abs{A}=\abs{B}=1+n+n(n+m)$.
    Conversely, if \PrbBNPShort returns an equally-sized partition, then $B$ cannot include
    any of $\Phi$ or $Y$, as this forces $\abs{B} > 1+n+n(n+m) > \abs{A}$.
    Since $Y \subseteq A$
    and  $\abs{X \cap A}=n$, exactly one of $\{x_i, \overline{x}_i\}$ is in $A$ for every $i$.
    Thus, in order for $\Phi$ to be a subset of $A$, $X \cap A$ must be a satisfying assignment for the \PrbSATShort instance.
\end{proof}

However, if $\max_{v\in V}d(v,L) \leq 1$
 (equivalently, $L$ is a ($1$-)dominating set of $G$),
then \PrbBNPShort becomes tractable.

\begin{theorem}\label{thm:runtime-bnp}
    \PrbBNPShort can be solved in $\bigo{|L|n^3}$ if\linebreak $\max_{v\in V}d(v,L) \leq 1$.
\end{theorem}

Our proof relies on Theorem~\ref{thm:min-var-square-relation} and the fact that the
problem is equivalent to \PrbSBA (\PrbSBAShort) which can be transformed into an instance of the maximum-cost minimum flow problem
(formal definition of \PrbSBAShort and proof in Appendix~\ref{sec:nbr-detail}).

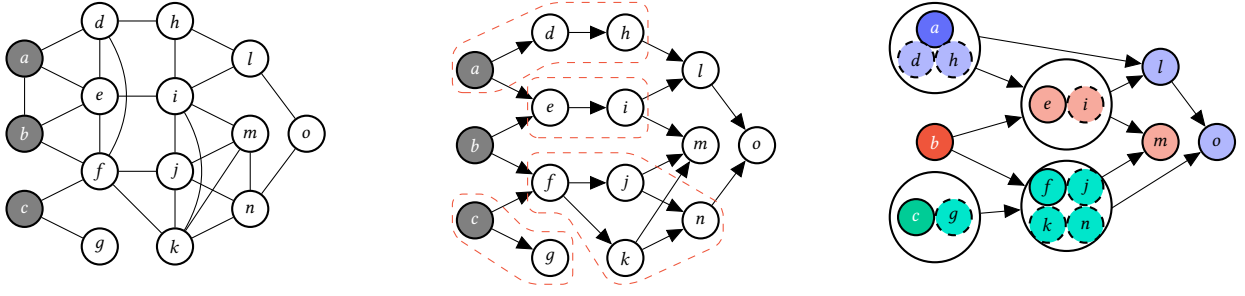
\begin{figure*}[ht]
    \centering

    \pgfdeclarelayer{bg}
    \pgfsetlayers{bg, main}

    \tikzstyle{bnode} = [circle, fill=gray, text=white, draw, thick, scale=0.8, minimum size=0.6cm, inner sep=1.5pt]
    \tikzstyle{wnode} = [circle, fill=white, text=black, draw, thick, scale=0.8, minimum size=0.6cm, inner sep=1.5pt]
    \tikzstyle{lnode} = [circle, fill=white, text=black, draw, thick, scale=0.8, minimum size=1.5cm, inner sep=1.5pt]

    \tikzstyle{directed} = [arrows=- triangle 45]

    \definecolor{color1}{RGB}{99,110,250}
    \definecolor{color4}{RGB}{177,183,253}
    \definecolor{color2}{RGB}{239,85,59}
    \definecolor{color5}{RGB}{247,170,157}
    \definecolor{color3}{RGB}{0,204,150}
    \definecolor{color6}{RGB}{0,230,203}

    \begin{minipage}[t]{.33\linewidth}
        \vspace{0pt}
        \centering
        \begin{tikzpicture}
            \node[bnode] (a) at (0, 2.5) {\textbf{$a$}};
            \node[bnode] (b) at (0, 1.5) {\textbf{$b$}};
            \node[bnode] (c) at (0, 0.5) {\textbf{$c$}};
            \node[wnode] (d) at (1, 3) {$d$};
            \node[wnode] (e) at (1, 2) {$e$};
            \node[wnode] (f) at (1, 1) {$f$};
            \node[wnode] (g) at (1, 0) {$g$};
            \node[wnode] (h) at (2, 3) {$h$};
            \node[wnode] (i) at (2, 2) {$i$};
            \node[wnode] (j) at (2, 1) {$j$};
            \node[wnode] (k) at (2, 0) {$k$};
            \node[wnode] (l) at (3, 2.5) {$l$};
            \node[wnode] (m) at (3, 1.5) {$m$};
            \node[wnode] (n) at (3, 0.5) {$n$};
            \node[wnode] (o) at (3.75, 1.5) {$o$};

            \draw (a) -- (b);
            \draw (d) -- (e); \draw (e) -- (f); \draw (d) to [out=-60,in=60] (f);
            \draw (h) -- (i); \draw (i) -- (j); \draw (j) -- (k);
            \draw (m) -- (n);
            \draw (a) -- (d); \draw (a) -- (e);
            \draw (b) -- (e); \draw (b) -- (f);
            \draw (c) -- (f); \draw (c) -- (g);
            \draw (d) -- (h);
            \draw (e) -- (i);
            \draw (f) -- (j);
            \draw (f) -- (k);
            \draw (h) -- (l);
            \draw (i) -- (l); \draw (i) -- (m); \draw (i) to [out=-60,in=60] (k);
            \draw (j) -- (m); \draw (j) -- (n);
            \draw (k) -- (m);
            \draw (k) -- (n);
            \draw (l) -- (o);
            \draw (n) -- (o);
        \end{tikzpicture}
    \end{minipage}%
    \begin{minipage}[t]{.34\linewidth}
        \vspace{0pt}
        \centering
        \begin{tikzpicture}
            \node[bnode] (a) at (0, 2.5) {\textbf{$a$}};
            \node[bnode] (b) at (0, 1.5) {\textbf{$b$}};
            \node[bnode] (c) at (0, 0.5) {\textbf{$c$}};
            \node[wnode] (d) at (1, 3)           {$d$};
            \node[wnode] (e) at (1, 2)           {$e$};
            \node[wnode] (f) at (1, 1)           {$f$};
            \node[wnode] (g) at (1, 0)           {$g$};
            \node[wnode] (h) at (2, 3)           {$h$};
            \node[wnode] (i) at (2, 2)           {$i$};
            \node[wnode] (j) at (2, 1)           {$j$};
            \node[wnode] (k) at (2, 0)           {$k$};
            \node[wnode] (l) at (3, 2.5)         {$l$};
            \node[wnode] (m) at (3, 1.5)         {$m$};
            \node[wnode] (n) at (3, 0.5)         {$n$};
            \node[wnode] (o) at (3.75, 1.5)      {$o$};

            \draw[directed] (a) -- (d); \draw[directed] (a) -- (e);
            \draw[directed] (b) -- (e); \draw[directed] (b) -- (f);
            \draw[directed] (c) -- (f); \draw[directed] (c) -- (g);
            \draw[directed] (d) -- (h);
            \draw[directed] (e) -- (i);
            \draw[directed] (f) -- (j);
            \draw[directed] (f) -- (k);
            \draw[directed] (h) -- (l);
            \draw[directed] (i) -- (l); \draw[directed] (i) -- (m);
            \draw[directed] (j) -- (m); \draw[directed] (j) -- (n);
            \draw[directed] (k) -- (m);
            \draw[directed] (k) -- (n);
            \draw[directed] (l) -- (o);
            \draw[directed] (n) -- (o);

            \draw[dashed, rounded corners, color=color2] (-0.3, 2.2) -- (0.3,2.2) -- (1, 2.6) -- (2.3, 2.6) -- (2.3, 3.4) -- (1, 3.4) -- (-0.3, 2.8) -- cycle;
            \draw[dashed, rounded corners, color=color2] (0.7, 1.6) -- (2.3, 1.6) -- (2.3, 2.4) -- (0.7, 2.4) -- cycle;
            \draw[dashed, rounded corners, color=color2] (0.7, 0.6) -- (1.2, 0.6) -- (1.8, -0.35) -- (2.1, -0.35) -- (3.3, 0.2) -- (3.3, 0.7) -- (2.1, 1.4) -- (0.7, 1.4) -- cycle;
            \draw[dashed, rounded corners, color=color2] (-0.3, 0.3) -- (0.7, -0.35) -- (1.3, -0.35) -- (1.3, 0.35) -- (0.7, 0.35) -- (0.2, 0.85) -- (-0.3, 0.85) -- cycle;
        \end{tikzpicture}
    \end{minipage}%
    \begin{minipage}[t]{.33\linewidth}
        \vspace{0pt}
        \centering
        \begin{tikzpicture}
            \node[lnode] (aa) at (0, 2.75) {};
            \node[lnode] (cc) at (0, 0.5) {};
            \node[lnode] (ee) at (1.75, 2) {};
            \node[lnode] (ff) at (1.75, 0.65) {};

            \node[bnode, fill=color1] (a) at (0, 3.0)     {\textbf{$a$}};
            \node[wnode, fill=color4, dashed] (d) at (-0.25, 2.6) {$d$};
            \node[wnode, fill=color4, dashed] (h) at (0.25, 2.6)  {$h$};

            \node[bnode, fill=color2] (b) at (0, 1.5) {\textbf{$b$}};

            \node[bnode, fill=color3] (c) at (-0.25, 0.5) {\textbf{$c$}};
            \node[wnode, fill=color6, dashed] (g) at (0.25, 0.5)  {$g$};

            \node[wnode, fill=color5] (e) at (1.5, 2)       {$e$};
            \node[wnode, fill=color5, dashed] (i) at (2, 2) {$i$};

            \node[wnode, fill=color6] (f) at (1.5, 0.9)         {$f$};
            \node[wnode, fill=color6, dashed] (j) at (2, 0.9)   {$j$};
            \node[wnode, fill=color6, dashed] (k) at (1.5, 0.4) {$k$};
            \node[wnode, fill=color6, dashed] (n) at (2, 0.4)   {$n$};

            \node[wnode, fill=color4] (l) at (3, 2.5) {$l$};
            \node[wnode, fill=color5] (m) at (3, 1.5) {$m$};

            \node[wnode, fill=color4] (o) at (3.75, 1.5) {$o$};

            \begin{pgfonlayer}{bg}
                \draw[directed] (h) -- (ee);
                \draw[directed] (b) -- (ee); \draw[directed] (b) -- (ff);
                \draw[directed] (cc) -- (ff);

                \draw[directed] (a) -- (l);
                \draw[directed] (i) -- (l); \draw[directed] (i) -- (m);
                \draw[directed] (j) -- (m);

                \draw[directed] (l) -- (o);
                \draw[directed] (n) -- (o);
            \end{pgfonlayer}
        \end{tikzpicture}
    \end{minipage}%
    \caption{%
    Visualization of neighborhood kernels.
    Given a graph $G$ with three landmarks $a, b, c$ (left),
     the neighborhood kernel $H$ of $G$ (center) identifies groups dotted in red
     which must all be assigned to the same landmark.
    At right, the compact neighborhood kernel of $G$ has one representative vertex
     (solid border) per bag;
     colors indicate an optimally balanced neighborhood partitioning.
    }
    \label{fig:nbr-kernel}
\end{figure*}

\subsection{Algorithms}
We now present several algorithms for computing neighborhood partitions
along with a quadratic programming formulation for \PrbBNPShort.
In all cases, we work with a preprocessed instance which we refer to as
a \textit{(compact) neighborhood kernel} which can be computed in linear time.
Our first two algorithms apply greedy strategies and are complemented by an exact
branch-and-bound algorithm based on ideas from an FPT algorithm on bounded-degree instances.
We defer some proofs of correctness and running time analyses to
Appendix~\ref{sec:alg-details}.

\subsubsection{Neighborhood Kernels}

By the definition of \PrbBNPShort, each non-landmark $v$ in a piece $A$ has a path of
length $d(v,L)$ to the landmark in $A$.
Thus, given a graph $G=(V,E)$ and landmarks $L \subseteq V$,
if an edge $e \in E$ is not in any shortest $v$-$L$ paths for all $v \in V$,
then we can safely remove $e$ from the original graph.
The underlying ideas of our preprocessing are orienting edges outward from the landmarks,
creating layers of vertices by distance from the closest landmark,
and removing unnecessary edges.
We call the output of this preprocessing a \textit{neighborhood kernel}.

\begin{definition}[Neighborhood Kernel]\label{def:nbr-knl}
    Given a graph $G=(V,E)$ and landmarks $L \subseteq V$,
    a \textbf{neighborhood kernel} $H$ is a digraph with vertex set $V$ and
    edge set $E'=\{(v,w): vw\in E,\ \ d(v,L) + 1 = d(w,L)\}$.
\end{definition}

By construction, for every non-landmark vertex $v$ in a neighborhood kernel,
 all shortest paths from $v$ to any landmark must include one of $v$'s in-neighbors in $H$,
 and thus $v$ must be in the same piece as one of its in-neighbors.
If $v$ has only one in-neighbor $w$, then $v$ and $w$ must be assigned to the same piece.
We encapsulate this idea in the following data reduction, noting that
 by definition, all landmarks are kept in the compact neighborhood kernel.
 \looseness-1

\begin{definition}[\small Compact Neighborhood Kernel\normalsize]\label{def:compact-nbr-knl}
    Given a neighborhood kernel $H$ for a graph $G=(V,E)$ and landmarks $L \subseteq V$,
     a \textbf{compact neighborhood kernel} $(H_c, \phi)$ is a pair consisting of
      a digraph $H_c=(V_c,E_c):=H/\mathcal{A}$ %
      and a map $\phi: V_c \to 2^V$
      defining its ``bags'',
      where the collection of bags $\mathcal{A}:=\{\phi(v): v \in V_c\}$
      is a partition of $V$.
     Additionally, we require the following conditions:
\begin{packed_item}
    \item Each vertex $v \in V_c$ is the representative for its bag,
     and must be the closest to $L$ in $G$ among $\phi(v)$.
    \item All bag members must be assigned to the same landmark in
     any valid neighborhood partitioning.
    \item $\phi(v)$ is maximal subject to these conditions.
\end{packed_item}
\end{definition}

We visualize this process in Figure~\ref{fig:nbr-kernel}.
In Appendix~\ref{sec:alg-details}, we describe
algorithms (\AlgNbrKnl and \AlgCmpNbrKnl) for
creating (compact) neighborhood kernels in $\bigo{n + m}$ time.
Finally, we note that the maximum degree cannot increase under these transformations
(Lemma~\ref{lem:compact-nbr-knl-in-degree}).
We use this fact to bound the running time of our branch-and-bound algorithm (Section~\ref{sec:branch-algorithm}).

\subsubsection{Heuristic Algorithms}

We developed two heuristics for \PrbBNPShort.
\AlgPrtWeight is a linear-time $\bigo{n + m}$ greedy algorithm that works
on the compact neighborhood kernel.
From landmarks, it traverses all bags in BFS order,
 assigning each bag to the smallest piece among viable candidates.

\begin{theorem}
    \AlgPrtWeight gives a valid neighborhood partitioning in $\bigo{n+m}$ time.
\end{theorem}

\begin{proof}
    We proceed by induction on distance from a landmark.
    As the base case, landmarks assign themselves and their bag members,
    which is the only valid option. For the inductive step, consider when $v$ is assigned.
    Because bags are processed in BFS order, $v$'s in-neighbors
    must have already been assigned to a valid landmark.
    Thus, extending the current assignment with any choice at $v$
    meets all requirements, and \AlgPrtWeight produces a valid partitioning.

    The running time is $\bigo{n+m}$ because
    this algorithm performs BFS, and a landmark's piece size
    can be updated in $\bigo{1}$.
\end{proof}


\AlgPrtLayer is a polynomial-time $\bigo{\abs{L} n^3}$ algorithm that is exact
when $L$ is a ($1$-)dominating set and heuristic otherwise.
It works on the neighborhood kernel and solves \PrbSBAShort at each layer starting
 from the one closest to landmarks
 (a layer $V_i$ is the set of vertices such that
 the distance to the closest landmark is $i$).

\begin{theorem}\label{thm:prt-layer-analysis}
    \AlgPrtLayer gives a neighborhood partitioning in $\bigo{\abs{L}n^3}$ time.
    Further, it gives an exact solution to \PrbBNPShort when $L$ is a dominating set of $G$.
\end{theorem}

\begin{proof}
    If each assigned piece is connected in the neighborhood kernel,
     then it is a valid neighborhood partitioning
     because vertices $V_i$ at layer $i$ are distance $i$ away from their closest landmarks,
     both in the original graph and their assigned piece.
    The algorithm examines vertices in level order emanating from the landmarks,
     and by the formulation of \PrbSBAShort, each of the resulting pieces must be connected.
    Also, from Theorem~\ref{thm:runtime-bnp},
    this algorithm gives an exact solution when $L$ is a dominating set.

    The running time (other than constructing a neighborhood kernel)
    is the sum of the time taken for \PrbSBAShort at each layer.
    At layer $i$, only the vertices in $V_i$ have choices
     (all vertices in earlier layers already have an assigned landmark).
    In a flow problem, we can efficiently preprocess already assigned vertices,
    thus resulting in time $\bigo{\abs{L}\abs{V_i}^3}$. Because of the fact
    $V=\bigcup_i{V_i}$, the total running time is
    $\sum_{i}{\bigo{\abs{L}\abs{V_i}^3}} = \bigo{\abs{L}n^3}$.
\end{proof}

\subsubsection{Branch-and-bound Algorithm}
\label{sec:branch-algorithm}

Before elaborating on our exact algorithm for \PrbBNPShort,
we show the following motivating result.

\begin{theorem}\label{thm:bnp-fpt}
    \PrbBNPShort is fixed parameter tractable (FPT) in graphs of bounded degree
    parameterized by $k$, the number of vertices equidistant from multiple landmarks.
\end{theorem}

\begin{proof}
    Given a graph $G=(V,E)$, landmarks $L \subseteq V$,
    and their compact neighborhood kernel $(H_c, \phi)$,
    let $\Delta$ be the maximum degree of $G$.
    Let $X$ be $V(H_c) \setminus L$.
    Since all vertices in $X$ are equidistant from multiple landmarks,
    $\abs{X} \leq k$.
    From Lemma~\ref{lem:compact-nbr-knl-in-degree},
    each vertex in $H_c$ has at most $\Delta$ in-neighbors.
    Every vertex in $X$ must be in the same piece as one of its in-neighbors
    for a valid neighborhood partitioning, giving at most $\Delta^k$ possible assignments.
    Other computations can be done in polynomial time in $n$,
    so the brute-force approach results in time $\bigo{\Delta^k n^{\bigo{1}}}$.
\end{proof}

The algorithm \AlgPrtBranch (documented and implemented in \cite{codebase}) reinforces
this idea by combining efficient base-case handling with naive branch-cut functionality.
For a base case, we apply \AlgPrtLayer if possible;
specifically, if there is only one layer $X_i$ left and all bags are of size 1,
\AlgPrtLayer gives an exact solution in $\bigo{\abs{L}n^3}$ time.
To obtain a lower-bound for a branch cut, we exploit the fact that our
partial solutions $\mathcal{A'}$ satisfy the following property:
any solution $\mathcal{A}$ extending $\mathcal{A'}$ has
$\sum_{A \in \mathcal{A}} \abs{A}^2 \geq
\left(\sum_{A' \in \mathcal{A'}} \abs{A'}^2\right) +
(n - \abs{\bigcup \mathcal{A'}})^2/\abs{L}$.
The total running time of this algorithm is $\bigo{\Delta^k \abs{L}n^3}$.

\begin{figure*}[ht]
    \centering
    \begin{minipage}[t]{.50\linewidth}
        \vspace{0pt}
        \centering
        \includegraphics[scale=0.12]{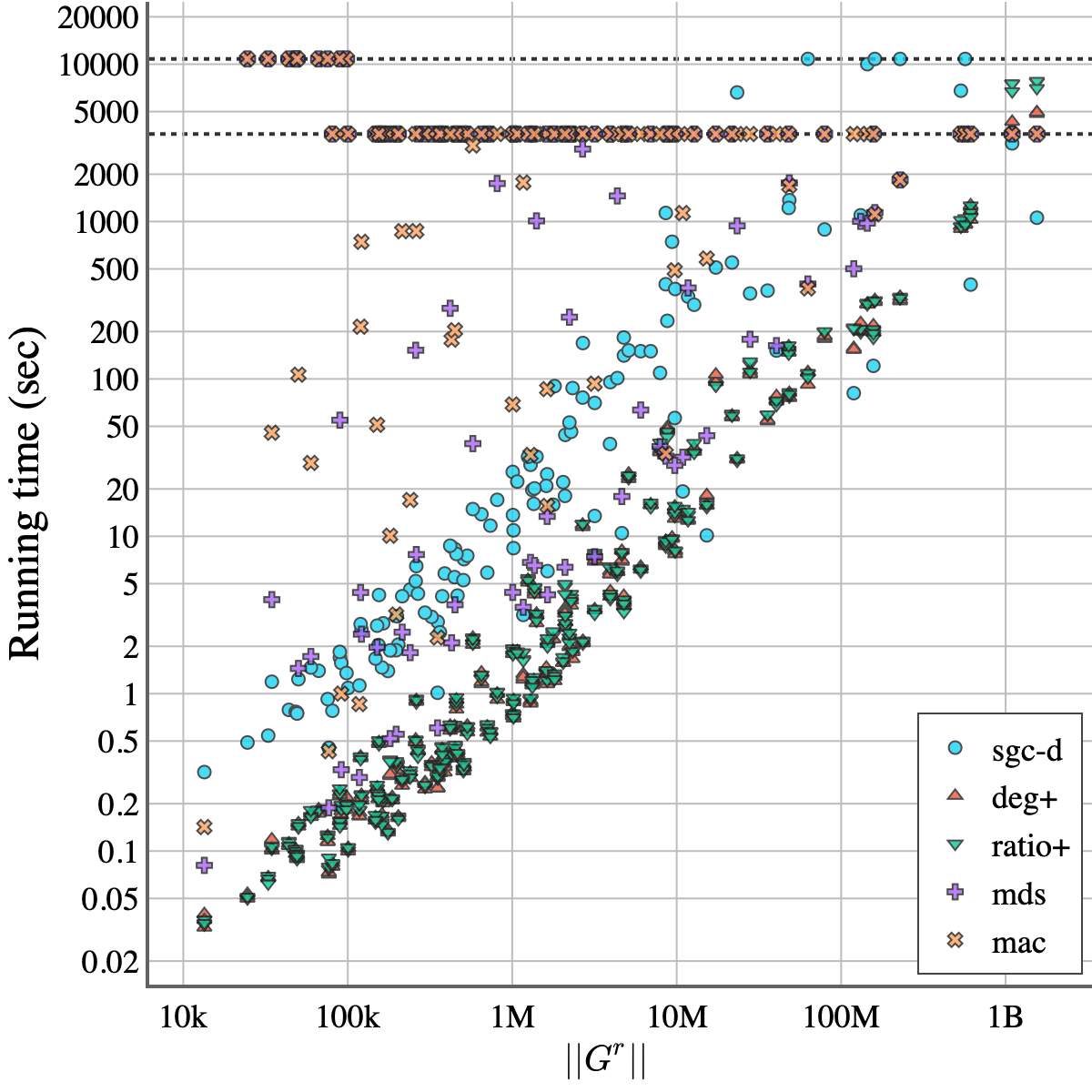}
    \end{minipage}%
    \begin{minipage}[t]{.50\linewidth}
        \vspace{0pt}
        \centering
        \includegraphics[scale=0.12]{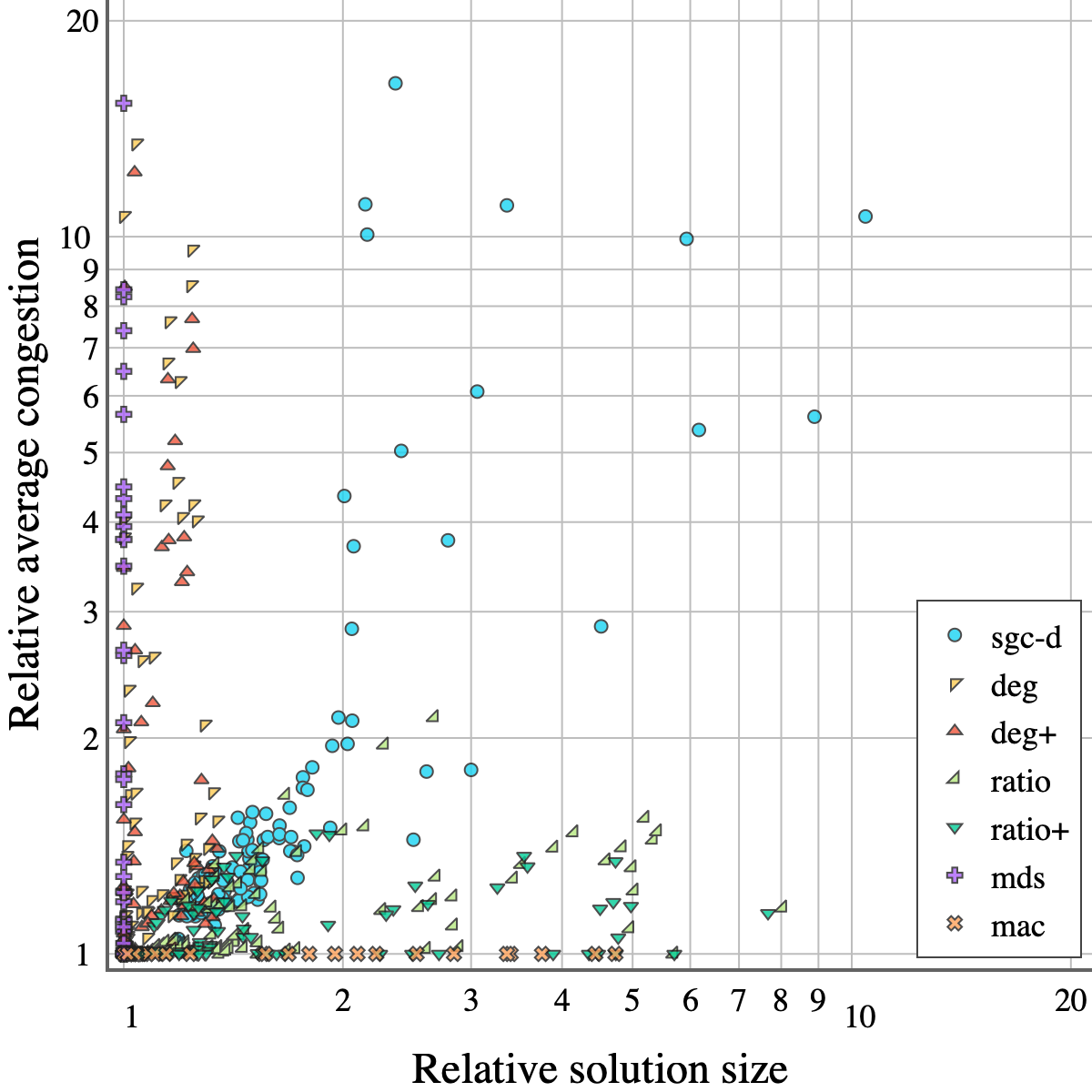}
    \end{minipage}%
    \caption{Experiment results of dominating set algorithms. 
    The left chart shows a log-scale distribution of runtimes sorted by $\norm{G^r}$
    for $r=1,2$.
    Time limits are depicted in dotted lines (1 hour and 3 hours). 
    The right chart shows a distribution of solution sizes and average congestions 
    relative to the best obtained for each experiment configuration.}
    \label{fig:exp-domset}
\end{figure*}

\subsubsection{Quadratic Programming}

Finally, we give a quadratic programming formulation.
Given a graph $G=(V,E)$ and landmarks $L \subseteq V$, let $(H_c=(V_c,E_c), \phi)$
be the compact neighborhood kernel of $G$ on $L$.
Let $x_{u,v} \in \{0,1\}$ be variables for all $u \in L, v \in V_c$;
set $x_{u,v}=1$ if and only if $v$ is assigned to landmark $u$.

\begin{LabelBox}{\AlgPrtQP}
\begin{talign*}
    \text{Minimize } &\sum_{u \in L} \left(\sum_{v \in V_c} \abs{\phi(v)} x_{u,v} \right)^2\\
    \text{Subject to } &\sum_{u \in L} x_{u,v} = 1 \quad\text{for all }v \in V_c,\\
    &\sum_{w \in N_{H_c}^{-}[v]} x_{u,w} \geq x_{u,v} \quad \text{for all } u \in L, v \in V_c \setminus L.\nonumber
\end{talign*}
\end{LabelBox}

By Theorem~\ref{thm:min-var-square-relation}, the objective function guarantees
minimum variance.
The first constraint enforces that every vertex must be assigned to exactly one landmark,
and the second constraint guarantees that every non-landmark must be in the same piece
as one of its in-neighbors.
The balanced neighborhood partitioning is given by
$\mathcal{A} = \{A_{u_1}, \ldots, A_{u_\ell}\}$
where $\displaystyle A_u = \bigcup_{v \in V_c: x_{u,v}=1} \phi(v)$.

\section{Experiments}
\label{sec:exp}

\begin{figure*}[ht]
    \centering
    \begin{minipage}[t]{.45\linewidth}
        \vspace{0pt}
        \centering
        \includegraphics[scale=0.16]{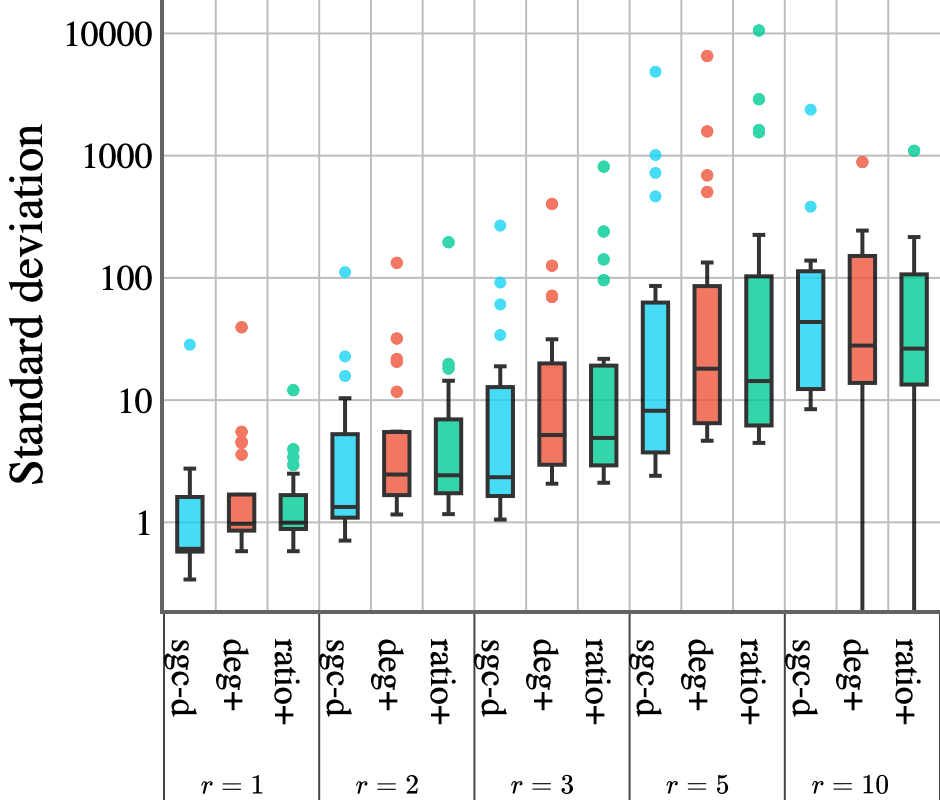}
    \end{minipage}%
    \begin{minipage}[t]{.55\linewidth}
        \vspace{0pt}
        \centering
        \includegraphics[scale=0.16]{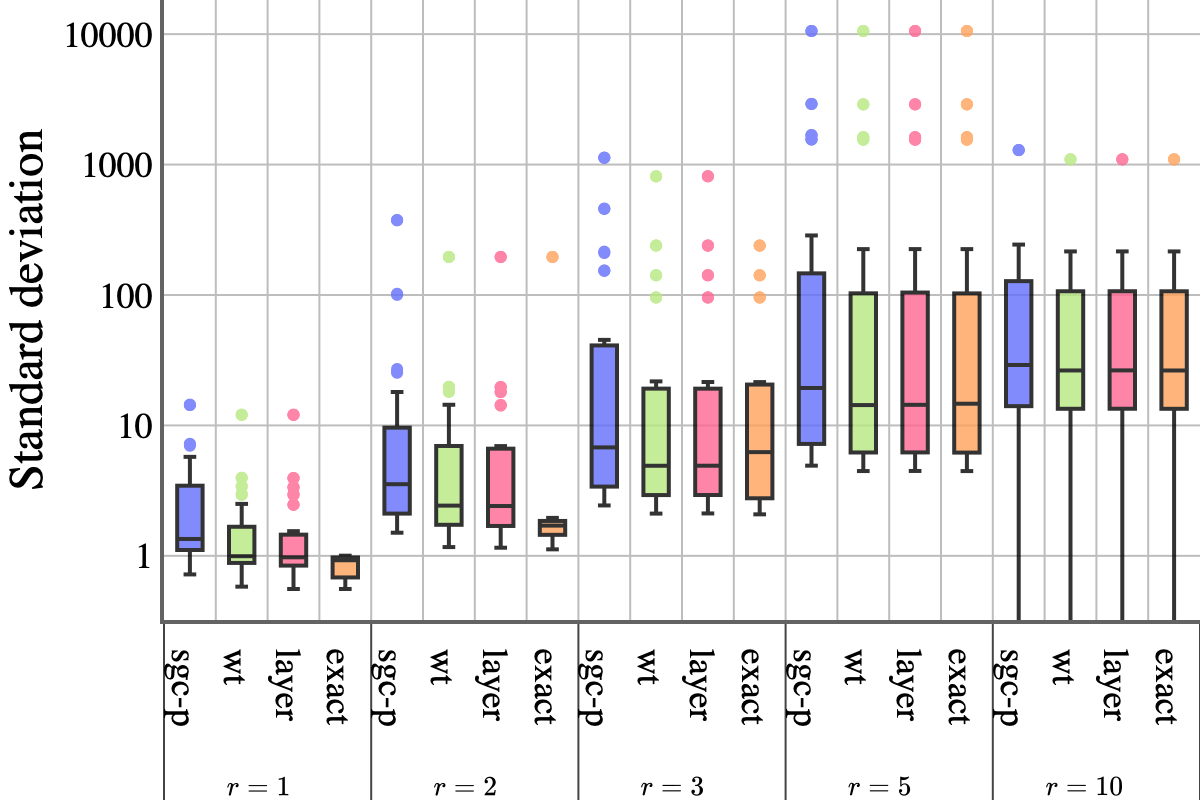}
    \end{minipage}
    \caption{Piece-size standard deviations from neighborhood partitioning algorithms
    on small/medium instances.
    At left, we compare the impact of dominator selection
    (with partitioning algorithm \AlgPrtWeight);
    at right, we fix the dominators (using \AlgRatioPlus) and vary the \PrbBNPShort solver.
    Exact solutions computed with \AlgPrtQPShort;
    results with timeouts are not shown.}
    \label{fig:exp-nbrprt}
\end{figure*}

To complement our theoretical results, we implemented and evaluated our algorithms
on a diverse corpus of networks including cDBGs constructed from real metagenomes,
small instances from DIMACS10~\cite{bader2013graph}, and graphs from~\cite{nguyen2020}.
We categorized each graph as \textit{small}, \textit{medium}, or \textit{large}
 based on the number of edges; all data, along with a summary table of statistics
 are available at~\cite{codebase}. All algorithms were implemented in C++ %
and experiments were conducted on a Linux machine
(details in Appendix~\ref{sec:exp-setup}).

\subsection{Sparse Dominating Sets}

To evaluate our algorithms for finding a low-congestion dominating set,
 we first tested the effectiveness of tie-breaking strategies by running
 \AlgDegree (denoted \AlgDegreeShort), \AlgRatio (\AlgRatioShort),
 \AlgDegreePlus (\AlgDegreePlusShort) as well as \AlgRatioPlus(\AlgRatioPlusShort)
 on small instances.
The algorithms with a tie-breaking strategy consistently outperformed those without one,
 finding smaller dominating sets in 78\% of the experiments and
 less-congested dominating sets in 93\% of them.
Running time was also improved by tie-breaking strategies in 55\% of instances,
 likely due to its proportionality to solution size.

Based on these findings, we restricted our attention to \AlgDegreePlusShort and
\AlgRatioPlusShort for the remaining experiments.
We evaluated them, along with \AlgDomSGC (\AlgDomSGCShort),\looseness-1
~\AlgMDS (\AlgMDSShort), and \AlgMAC (\AlgMACShort) on the full corpus with radii $\{1, 2\}$
(plus radii $\{3, 5, 10\}$ on small and medium instances). We used a default 3-hour timeout,
reduced to 1 hour for \AlgMDSShort and \AlgMACShort for radius $>1$ and medium/large instances.

Figure~\ref{fig:exp-domset} (left) shows the distribution of running times sorted
by $\norm{G^r}$, the number of edges in the $r$-th power of $G$.
We observe a linear trend for \AlgDomSGCShort, \AlgDegreePlusShort, and \AlgRatioPlusShort,
establishing efficiency.
While the exact algorithms \AlgMDSShort and \AlgMACShort are prone to timing out, they did
finish on some larger instances. We hypothesize this success is directly related to reduction
rules in the ILP solver related to vertices with degree $1$ and $2$.

To evaluate solution quality, we plotted the relationship between the solution size and
average congestion, both relative to the best-known (Figure~\ref{fig:exp-domset}, right).
As might be expected, \AlgDegreePlusShort and \AlgMDSShort find smaller, more congested
dominating sets, while \AlgRatioPlusShort and \AlgMACShort find larger, sparser ones.
The algorithm used in the prior metagenomic analysis (\AlgDomSGCShort) intermediates between
the two, but often at the cost of being much larger or more congested.\looseness-1

\subsection{Balanced Neighborhood Partitioning}

Turning to the problem of generating uniformly-sized pieces around a set of dominators,
we tested the scalability and solution quality of the algorithms for \PrbBNPShort
as well as the impact of choosing smaller versus sparser dominating sets.
To this end, we ran \AlgPrtSGC (\AlgPrtSGCShort),
\AlgPrtWeight (\AlgPrtWeightShort), \AlgPrtLayer (\AlgPrtLayerShort),
\AlgPrtBranch \linebreak(\AlgPrtBranchShort), and \AlgPrtQP (\AlgPrtQPShort)
using dominating sets produced by \AlgDomSGCShort, \AlgDegreePlusShort, and \AlgRatioPlusShort.
All runs were subject to a 1-hour timeout.

We first evaluated the running time and timeout rate for all \PrbBNPShort
algorithms (Table~\ref{tab:exp-nbr-runtime}) on the corpus of small instances.
The algorithm \AlgPrtWeightShort was consistently fastest, followed by \AlgPrtSGCShort;
both had no timeouts.
The polynomial-time \AlgPrtLayerShort was slower than the linear-time approaches but still
completed nearly all small and medium instances. It is particularly notable that when $r=1$,
it successfully output optimal solutions, which \AlgPrtQPShort could not find.
We also observed that \AlgPrtQPShort's performance improves when given a cDBG and dominators
from \AlgDegreePlusShort or \AlgRatioPlusShort.
Lastly, \AlgPrtBranchShort was unable to finish with larger radii,
even on these small instances, eliminating it from use in additional experiments.

\noindent\begin{minipage}{\linewidth}
    \vspace*{1em}
    \centering
    \makebox[\linewidth]{
        \centering
        \footnotesize
        \begin{tabular}{|c|c|r|r|}
            \hline
            Config & Algorithm & Rel. Runtime & \%Timeout \\
            \hline
            \multirow{5}{*}{All}
            & \AlgPrtWeightShort &  1.0 & --- \\
            & \AlgPrtSGCShort & 28.1 & --- \\
            & \AlgPrtLayerShort    & 49.8 & --- \\
            & \AlgPrtBranchShort  & 135.6 & 75.0\% \\
            & \AlgPrtQPShort     & 667.5  & 36.1\% \\
            \hline
            & \AlgPrtWeightShort &  1.0 & --- \\
            cDBG & \AlgPrtSGCShort & 35.4 & --- \\
            + & \AlgPrtLayerShort & 28.9 & --- \\
            \AlgRatioPlusShort & \AlgPrtBranchShort  & 71.8 & 80.0\% \\
            & \AlgPrtQPShort     & 149.3  & 13.3\% \\
            \hline
        \end{tabular}
    }

    \captionsetup{hypcap=false}
    \captionof{table}{Average running time relative to \AlgPrtWeight (\AlgPrtWeightShort)
     (Rel. Runtime) and timeout rate for \PrbBNPShort algorithms
     on the corpus of small instances.}
    \label{tab:exp-nbr-runtime}
    \vspace*{1em}
\end{minipage}

\vspace*{-1.5em}

For both remaining experiments, we ran on all small and medium instances.
To assess the impact of dominator selection on the variance of piece sizes, we restrict our
attention to \AlgPrtWeightShort and measure the standard deviations of
the piece sizes (Figure~\ref{fig:exp-nbrprt} (left)).
In general, dominating sets from \AlgRatioPlusShort tend to result in more balanced
neighborhood partitionings than \AlgDegreePlusShort, with some outliers at larger radii.
We note that while \AlgPrtSGCShort achieves the least variance (except at $r=10$),
this is in part due to having larger dominating sets (thus a smaller mean piece size), exposing a limitation of this metric.

Finally, we compare the solution quality of our greedy algorithms for neighborhood
partitioning using fixed dominating sets found by \AlgRatioPlusShort.
In Figure~\ref{fig:exp-nbrprt} (right), we see that \AlgPrtWeightShort achieves more balanced
pieces than \AlgPrtSGCShort at all radii while keeping the same asymptotic
running time. The more expensive \AlgPrtLayerShort performs even better, achieving nearly optimal results.
With larger radii (e.g. $r \geq 3$), \AlgPrtWeightShort performs as well as \AlgPrtLayerShort
 does.

\subsection{Metagenome Neighborhood Queries}

Given the motivation for this work is improving the metagenomics analysis pipeline
from~\cite{Brown2020}, we also assessed the impact of our techniques in this setting.
Specifically, we verified that using sparse dominating sets and balanced neighborhood partitionings
from our algorithms does not significantly degrade the containment and similarity of neighborhood
queries in \textsf{HuSB1}, a large real-world metagenome \cite{hu2016genome}.
We used Brown et al.'s replication pipeline\footnote{https://github.com/ctb/2020-rerun-hu-s1}
and compared our algorithms (\AlgDegreePlusShort and \AlgRatioPlusShort for \PrbMCDSShort;
\AlgPrtWeightShort for \PrbBNPShort) to those used in~\cite{Brown2020}
(\AlgDomSGCShort and \AlgPrtSGCShort, respectively).
While~\cite{Brown2020} restricted their attention to radius 1, we also include results from
radius 2, as this is of interest in ongoing related work by Brown et al.\looseness-1

Before running neighborhood queries, we evaluated the piece size distributions as shown in
Figure~\ref{fig:exp-query} (left, center). We observe that independent of dominator
selection, \AlgPrtWeightShort successfully reduced the piece size variance.
When comparing dominating sets, we see a trend consistent with other experiments.
\AlgDomSGCShort chooses many dominators, resulting in smaller more uniform piece sizes.
\AlgDegreePlusShort has the fewest dominators and higher variance, larger pieces and
\AlgRatioPlusShort gives larger, balanced pieces.

\begin{figure*}[ht]
    \centering
    \begin{minipage}[m]{.5\linewidth}
        \vspace{0pt}
        \centering
        \includegraphics[scale=0.16]{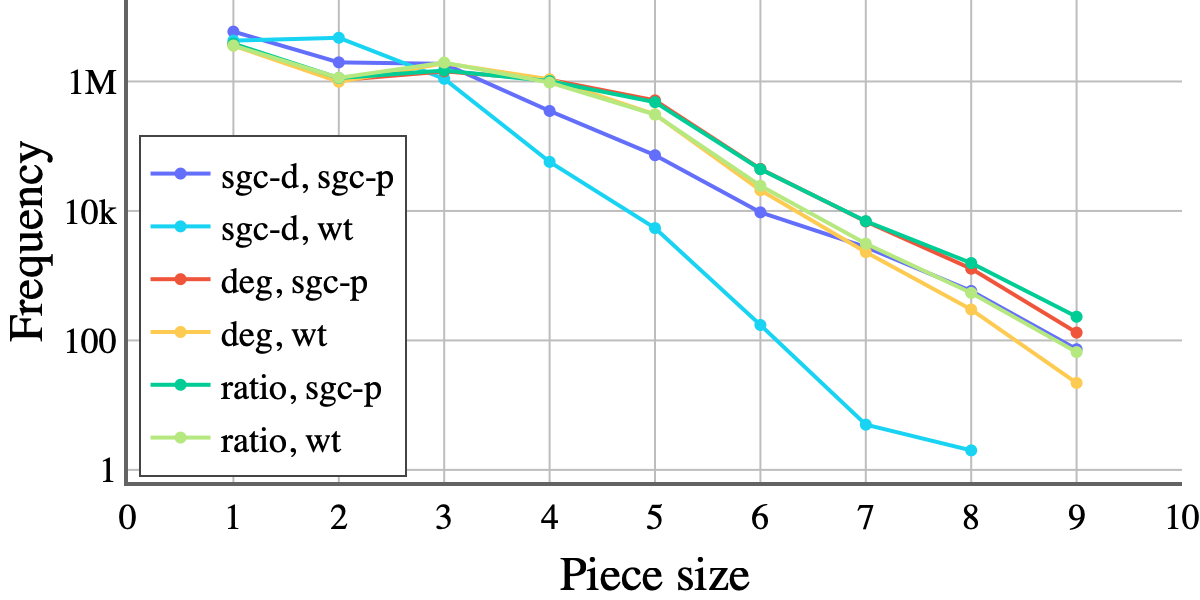} \\
        \includegraphics[scale=0.16]{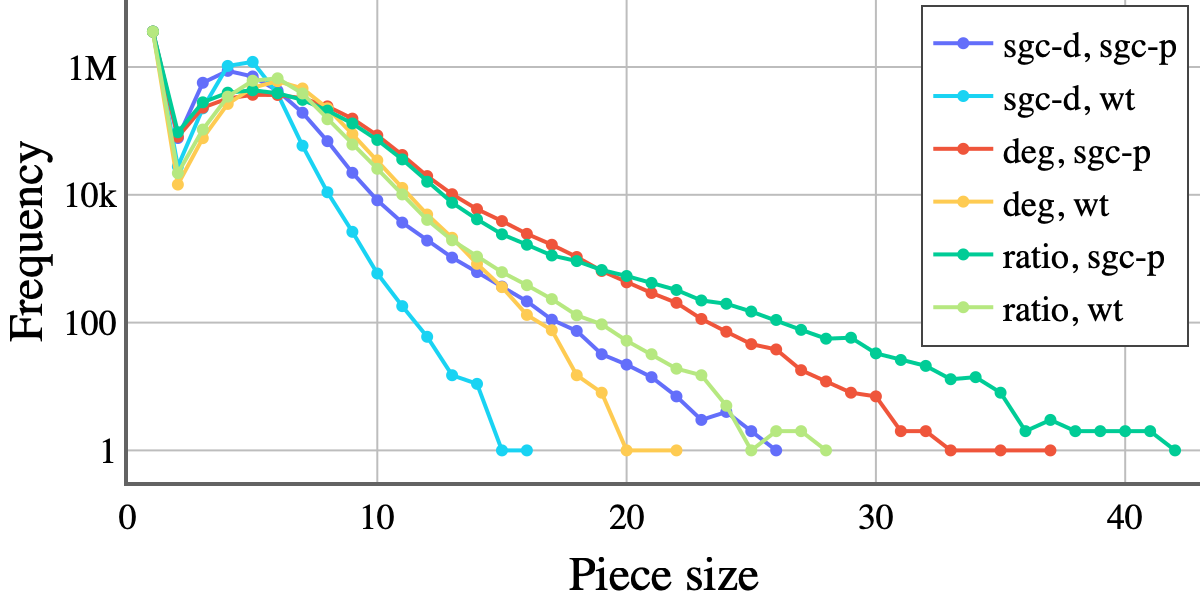}
    \end{minipage}%
%
%
    \begin{minipage}[m]{.49\linewidth}
        \vspace{0pt}
        \centering
        \includegraphics[scale=0.20]{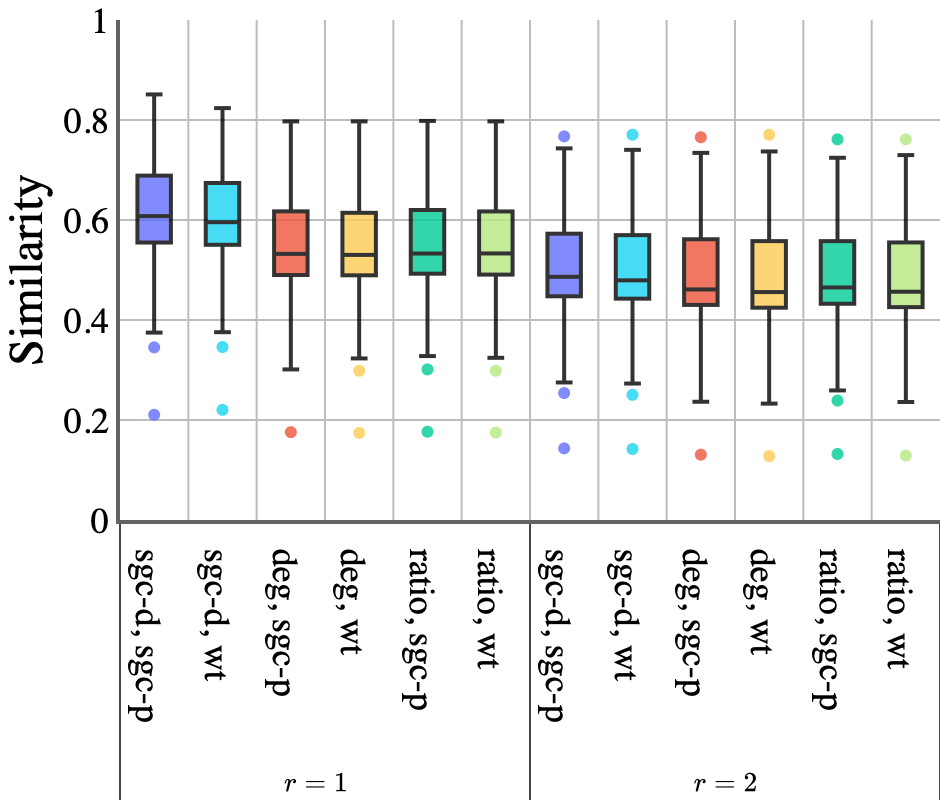}
    \end{minipage}%
    \caption{Experiment results of neighborhood queries. The left panel shows
    log-scale distributions of the resulting piece sizes
    for radius $1$ (top) and $2$ (bottom).
    The right panel shows the Jaccard similarity of the queries and their neighborhoods.}
    \label{fig:exp-query}
\end{figure*}

Finally, we re-ran the neighborhood queries from ~\cite{Brown2020}.
While the containment was completely preserved,
the similarities\footnote{Jaccard similarity to original query} exhibited mild variation,
shown in Figure~\ref{fig:exp-query} (right). While there is a slight overall reduction
(partly explainable by having fewer pieces), the difference is marginal at $r=2$.
These preliminary results indicate that our techniques can significantly improve the balance
of piece sizes (leading to more efficient downstream analysis when a hierarchy of dominating
sets is constructed) without significantly degrading fidelity of queries.
Additional experiments are needed to better assess the impact of this work in the
metagenomic setting.
%

\section{Conclusions \& Future Work}

This paper tackles two problems arising in a recent approach to metagenome analysis:
finding ($r$-)dominating sets which minimize the number of vertices with multiple closest
dominators, and partitioning the vertices of a graph into connected pieces around a set of landmarks
minimizing variance of the piece-sizes while guaranteeing every vertex is assigned to some
closest landmark.

We formalize the first using the notion of \textit{congestion},
 and show that finding minimum congestion $r$-dominating sets (\PrbMCDSShort) is NP-hard.
We introduce linear-time algorithms for finding low-congestion dominating sets and evaluate
 their effectiveness on a large corpus of real-world data,
 showing trade-offs between solution size and average congestion.
It remains open whether there is a constant upper bound on the minimum average congestion
 (the largest value observed in our experiments was 3.7351);
 further, we believe the approximation bound on \PrbMCDSShort is not tight.
We are intrigued by the connection of this problem to \PrbPCE; this could be a fruitful direction for future work.

Turning to the partitioning problem (\PrbBNPShort), we show that at radius $1$,
flow-based techniques give exact solutions in polynomial time, but the problem becomes
NP-hard even in very restricted cases as soon as $r \geq 2$.
Our heuristics, however, produce nearly optimal results in our experiments.
Further, we show that using sparse dominating sets (such as those from \AlgRatioPlusShort)
improves both the running time and solution quality of algorithms for \PrbBNPShort.
A natural question is whether there are alternative measures of partition balance in
real-world data which avoid the inverse relationship between dominating set size
and the variance of the piece size distribution.\looseness-1

Finally, we integrate our algorithms into the metagenomic analysis pipeline
used in~\cite{Brown2020} and demonstrate that on a large real metagenome (\textsf{HuSB1}),
sparse dominating sets and balanced neighborhood partitionings reduce piece size variability
relative to the prior approaches without significantly degrading the fidelity of neighborhood
queries.
It remains to see how these preliminary results extend to other metagenomic datasets and
queries.

\begin{acks}
  Special thanks to Rachel Walker for early discussions on notions of
 sparse dominating sets \& related terminology, and to Titus Brown \& Taylor Reiter
 for assistance with metagenomic datasets and the \sgc software.
This work was supported by grant GBMF4560 to Blair D. Sullivan
 from the Gordon and Betty Moore Foundation. \looseness-1

\end{acks}

\bibliographystyle{ACM-Reference-Format}
\bibliography{ms}


\begin{thebibliography}{19}


\ifx \showCODEN    \undefined \def \showCODEN     #1{\unskip}     \fi
\ifx \showDOI      \undefined \def \showDOI       #1{#1}\fi
\ifx \showISBNx    \undefined \def \showISBNx     #1{\unskip}     \fi
\ifx \showISBNxiii \undefined \def \showISBNxiii  #1{\unskip}     \fi
\ifx \showISSN     \undefined \def \showISSN      #1{\unskip}     \fi
\ifx \showLCCN     \undefined \def \showLCCN      #1{\unskip}     \fi
\ifx \shownote     \undefined \def \shownote      #1{#1}          \fi
\ifx \showarticletitle \undefined \def \showarticletitle #1{#1}   \fi
\ifx \showURL      \undefined \def \showURL       {\relax}        \fi
\providecommand\bibfield[2]{#2}
\providecommand\bibinfo[2]{#2}
\providecommand\natexlab[1]{#1}
\providecommand\showeprint[2][]{arXiv:#2}

\bibitem[Alon et~al\mbox{.}(2006)]%
        {alon2006algorithmic}
\bibfield{author}{\bibinfo{person}{Noga Alon}, \bibinfo{person}{Dana
  Moshkovitz}, {and} \bibinfo{person}{Shmuel Safra}.}
  \bibinfo{year}{2006}\natexlab{}.
\newblock \showarticletitle{Algorithmic construction of sets for
  k-restrictions}.
\newblock \bibinfo{journal}{\emph{ACM Transactions on Algorithms (TALG)}}
  \bibinfo{volume}{2}, \bibinfo{number}{2} (\bibinfo{year}{2006}),
  \bibinfo{pages}{153--177}.
\newblock


\bibitem[Bader et~al\mbox{.}(2013)]%
        {bader2013graph}
\bibfield{author}{\bibinfo{person}{David~A Bader}, \bibinfo{person}{Henning
  Meyerhenke}, \bibinfo{person}{Peter Sanders}, {and} \bibinfo{person}{Dorothea
  Wagner}.} \bibinfo{year}{2013}\natexlab{}.
\newblock \bibinfo{booktitle}{\emph{Graph partitioning and graph clustering}}.
  Vol.~\bibinfo{volume}{588}.
\newblock \bibinfo{publisher}{American Mathematical Society, Providence, RI}.
\newblock


\bibitem[Brown et~al\mbox{.}(2020)]%
        {Brown2020}
\bibfield{author}{\bibinfo{person}{C.~Titus Brown}, \bibinfo{person}{Dominik
  Moritz}, \bibinfo{person}{Michael~P. O'Brien}, \bibinfo{person}{Felix Reidl},
  \bibinfo{person}{Taylor Reiter}, {and} \bibinfo{person}{Blair~D. Sullivan}.}
  \bibinfo{year}{2020}\natexlab{}.
\newblock \showarticletitle{Exploring neighborhoods in large metagenome
  assembly graphs using spacegraphcats reveals hidden sequence diversity}.
\newblock \bibinfo{journal}{\emph{Genome Biology}} \bibinfo{volume}{21},
  \bibinfo{number}{1} (\bibinfo{date}{06 Jul} \bibinfo{year}{2020}),
  \bibinfo{pages}{164}.
\newblock
\showISSN{1474-760X}


\bibitem[Chikhi et~al\mbox{.}(2016)]%
        {10.1093/bioinformatics/btw279}
\bibfield{author}{\bibinfo{person}{Rayan Chikhi}, \bibinfo{person}{Antoine
  Limasset}, {and} \bibinfo{person}{Paul Medvedev}.}
  \bibinfo{year}{2016}\natexlab{}.
\newblock \showarticletitle{{Compacting de Bruijn graphs from sequencing data
  quickly and in low memory}}.
\newblock \bibinfo{journal}{\emph{Bioinformatics}} \bibinfo{volume}{32},
  \bibinfo{number}{12} (\bibinfo{date}{06} \bibinfo{year}{2016}),
  \bibinfo{pages}{i201--i208}.
\newblock
\showISSN{1367-4803}


\bibitem[Dvořák(2013)]%
        {dvorak2013constantfactor}
\bibfield{author}{\bibinfo{person}{Zdeněk Dvořák}.}
  \bibinfo{year}{2013}\natexlab{}.
\newblock \showarticletitle{Constant-factor approximation of the domination
  number in sparse graphs}.
\newblock \bibinfo{journal}{\emph{European Journal of Combinatorics}}
  \bibinfo{volume}{34}, \bibinfo{number}{5} (\bibinfo{year}{2013}),
  \bibinfo{pages}{833--840}.
\newblock
\showISSN{0195-6698}


\bibitem[Einarson and Reidl(2020)]%
        {einarson2020general}
\bibfield{author}{\bibinfo{person}{Carl Einarson} {and} \bibinfo{person}{Felix
  Reidl}.} \bibinfo{year}{2020}\natexlab{}.
\newblock \showarticletitle{{A General Kernelization Technique for Domination
  and Independence Problems in Sparse Classes}}. In
  \bibinfo{booktitle}{\emph{15th International Symposium on Parameterized and
  Exact Computation (IPEC 2020)}} \emph{(\bibinfo{series}{Leibniz International
  Proceedings in Informatics (LIPIcs)}, Vol.~\bibinfo{volume}{180})}.
  \bibinfo{pages}{11:1--11:15}.
\newblock
\showISBNx{978-3-95977-172-6}
\showISSN{1868-8969}


\bibitem[Ford and Fulkerson(1956)]%
        {ford1956maximal}
\bibfield{author}{\bibinfo{person}{Lester~Randolph Ford} {and}
  \bibinfo{person}{Delbert~R Fulkerson}.} \bibinfo{year}{1956}\natexlab{}.
\newblock \showarticletitle{Maximal flow through a network}.
\newblock \bibinfo{journal}{\emph{Canadian journal of Mathematics}}
  \bibinfo{volume}{8} (\bibinfo{year}{1956}), \bibinfo{pages}{399--404}.
\newblock


\bibitem[Garey and Johnson(1979)]%
        {garey1979computers}
\bibfield{author}{\bibinfo{person}{Michael~R Garey} {and}
  \bibinfo{person}{David~S Johnson}.} \bibinfo{year}{1979}\natexlab{}.
\newblock \bibinfo{booktitle}{\emph{Computers and intractability}}.
  Vol.~\bibinfo{volume}{174}.
\newblock \bibinfo{publisher}{freeman San Francisco}.
\newblock


\bibitem[Hu et~al\mbox{.}(2016)]%
        {hu2016genome}
\bibfield{author}{\bibinfo{person}{Ping Hu}, \bibinfo{person}{Lauren Tom},
  \bibinfo{person}{Andrea Singh}, \bibinfo{person}{Brian~C. Thomas},
  \bibinfo{person}{Brett~J. Baker}, \bibinfo{person}{Yvette~M. Piceno},
  \bibinfo{person}{Gary~L. Andersen}, \bibinfo{person}{Jillian~F. Banfield},
  {and} \bibinfo{person}{Nicole Dubilier}.} \bibinfo{year}{2016}\natexlab{}.
\newblock \showarticletitle{Genome-Resolved Metagenomic Analysis Reveals Roles
  for Candidate Phyla and Other Microbial Community Members in Biogeochemical
  Transformations in Oil Reservoirs}.
\newblock \bibinfo{journal}{\emph{mBio}} \bibinfo{volume}{7},
  \bibinfo{number}{1} (\bibinfo{year}{2016}), \bibinfo{pages}{e01669--15}.
\newblock


\bibitem[Jaffke et~al\mbox{.}(2019)]%
        {JAFFKE2019216}
\bibfield{author}{\bibinfo{person}{Lars Jaffke}, \bibinfo{person}{O joung
  Kwon}, \bibinfo{person}{Torstein~J.F. Strømme}, {and}
  \bibinfo{person}{Jan~Arne Telle}.} \bibinfo{year}{2019}\natexlab{}.
\newblock \showarticletitle{Mim-width III. Graph powers and generalized
  distance domination problems}.
\newblock \bibinfo{journal}{\emph{Theoretical Computer Science}}
  \bibinfo{volume}{796} (\bibinfo{year}{2019}), \bibinfo{pages}{216--236}.
\newblock
\showISSN{0304-3975}


\bibitem[Kratochvíl(1994)]%
        {KRATOCHVIL1994191}
\bibfield{author}{\bibinfo{person}{Jan Kratochvíl}.}
  \bibinfo{year}{1994}\natexlab{}.
\newblock \showarticletitle{Regular codes in regular graphs are difficult}.
\newblock \bibinfo{journal}{\emph{Discrete Mathematics}} \bibinfo{volume}{133},
  \bibinfo{number}{1} (\bibinfo{year}{1994}), \bibinfo{pages}{191--205}.
\newblock
\showISSN{0012-365X}


\bibitem[Livingston and Stout(1997)]%
        {marilynn1997}
\bibfield{author}{\bibinfo{person}{Marilynn Livingston} {and}
  \bibinfo{person}{Q. Stout}.} \bibinfo{year}{1997}\natexlab{}.
\newblock \showarticletitle{Perfect Dominating Sets}.
\newblock \bibinfo{journal}{\emph{Congressus Numerantium}}
  \bibinfo{volume}{79} (\bibinfo{date}{08} \bibinfo{year}{1997}).
\newblock


\bibitem[Mizutani et~al\mbox{.}(2022)]%
        {codebase}
\bibfield{author}{\bibinfo{person}{Yosuke Mizutani}, \bibinfo{person}{Annie
  Staker}, {and} \bibinfo{person}{Blair~D. Sullivan}.}
  \bibinfo{year}{2022}\natexlab{}.
\newblock \bibinfo{title}{Accompanying source code}.
\newblock
  \bibinfo{howpublished}{\url{https://github.com/TheoryInPractice/sparsedomsets}}.
\newblock


\bibitem[Nguyen et~al\mbox{.}(2020)]%
        {nguyen2020}
\bibfield{author}{\bibinfo{person}{Minh Nguyen}, \bibinfo{person}{Minh Hà},
  \bibinfo{person}{Diep Nguyen}, {and} \bibinfo{person}{The Tran}.}
  \bibinfo{year}{2020}\natexlab{}.
\newblock \showarticletitle{Solving the k-dominating set problem on very
  large-scale networks}.
\newblock \bibinfo{journal}{\emph{Comp. Social Networks}}  \bibinfo{volume}{7}
  (\bibinfo{date}{07} \bibinfo{year}{2020}).
\newblock


\bibitem[Quince et~al\mbox{.}(2017)]%
        {quince2017shotgun}
\bibfield{author}{\bibinfo{person}{Christopher Quince}, \bibinfo{person}{Alan~W
  Walker}, \bibinfo{person}{Jared~T Simpson}, \bibinfo{person}{Nicholas~J
  Loman}, {and} \bibinfo{person}{Nicola Segata}.}
  \bibinfo{year}{2017}\natexlab{}.
\newblock \showarticletitle{Shotgun metagenomics, from sampling to analysis}.
\newblock \bibinfo{journal}{\emph{Nature biotechnology}} \bibinfo{volume}{35},
  \bibinfo{number}{9} (\bibinfo{year}{2017}), \bibinfo{pages}{833--844}.
\newblock


\bibitem[Raz and Safra(1997)]%
        {raz1997sub}
\bibfield{author}{\bibinfo{person}{Ran Raz} {and} \bibinfo{person}{Shmuel
  Safra}.} \bibinfo{year}{1997}\natexlab{}.
\newblock \showarticletitle{A sub-constant error-probability low-degree test,
  and a sub-constant error-probability PCP characterization of NP}. In
  \bibinfo{booktitle}{\emph{Proceedings of the twenty-ninth annual ACM
  symposium on Theory of computing}}. \bibinfo{pages}{475--484}.
\newblock


\bibitem[Rossi and Ahmed(2015)]%
        {nr}
\bibfield{author}{\bibinfo{person}{Ryan~A. Rossi} {and}
  \bibinfo{person}{Nesreen~K. Ahmed}.} \bibinfo{year}{2015}\natexlab{}.
\newblock \showarticletitle{The Network Data Repository with Interactive Graph
  Analytics and Visualization}. In \bibinfo{booktitle}{\emph{AAAI}}.
\newblock


\bibitem[Slater(1976)]%
        {slater1976}
\bibfield{author}{\bibinfo{person}{Peter~J. Slater}.}
  \bibinfo{year}{1976}\natexlab{}.
\newblock \showarticletitle{R-Domination in Graphs}.
\newblock \bibinfo{journal}{\emph{J. ACM}} \bibinfo{volume}{23},
  \bibinfo{number}{3} (\bibinfo{date}{July} \bibinfo{year}{1976}),
  \bibinfo{pages}{446–450}.
\newblock
\showISSN{0004-5411}


\bibitem[Wang et~al\mbox{.}(2012)]%
        {wang2012ptas}
\bibfield{author}{\bibinfo{person}{Zhong Wang}, \bibinfo{person}{Wei Wang},
  \bibinfo{person}{Joon-Mo Kim}, \bibinfo{person}{Bhavani Thuraisingham}, {and}
  \bibinfo{person}{Weili Wu}.} \bibinfo{year}{2012}\natexlab{}.
\newblock \showarticletitle{PTAS for the minimum weighted dominating set in
  growth bounded graphs}.
\newblock \bibinfo{journal}{\emph{Journal of Global Optimization}}
  \bibinfo{volume}{54}, \bibinfo{number}{3} (\bibinfo{year}{2012}),
  \bibinfo{pages}{641--648}.
\newblock


\end{thebibliography}

\appendix

\newpage
\section{Theoretical Results}
\label{sec:theory}

\subsection{Neighborhood Partitioning Proofs}
\label{sec:nbr-detail}

First, we prove Theorem~\ref{thm:min-var-square-relation},
 which relates the minimum variance to the minimum sum of squares.

\begin{proof}[Proof of Theorem~\ref{thm:min-var-square-relation}]
    Let $X=\{x_1,\ldots,x_\ell\}$ be the piece sizes of $\mathcal{A}$.
    By definition of population variance,
    $\text{Var}(X)=$\linebreak $\frac{1}{\abs{X}}\sum_{i=1}^\ell \left(x_i - \mu\right)^2=
    \left(\frac{1}{\abs{X}}\sum_{i=1}^\ell x_i^2 \right) - \mu^2$,
    where $\mu=\frac{1}{\abs{X}}\sum_{i=1}^{\ell}x_i=\abs{V}/\ell$.
    Since $\mu$ is independent of $\mathcal{A}$,
    the minimum variance is always achieved with the minimum square sum.
\end{proof}

Next, we define \PrbSBA, which can be viewed as a generalization of an assignment problem
where every \textit{agent} is required to have the same number of \textit{tasks}.

\begin{ProblemBox}{\PrbSBA (\PrbSBAShort)}
    \Input & Agents $A$, tasks $T$, and a binary relation $R \subseteq A \times T$.\\
    \Prob  & Find a function $\phi: T \rightarrow A$ such that $(\phi(t), t) \in R$
    for every $t \in T$, and the population variance of the number of assigned tasks
    for each agent $\text{Var} (\{ \abs{\phi^{-1}(a)}  : a \in A\})$ is minimized.\\
\end{ProblemBox}
We use the following results to prove Theorem~\ref{thm:runtime-bnp}.
\vspace{1em}
\begin{theorem}\label{thm:runtime-bap}
    \PrbSBAShort can be solved in time $\mathcal{O}(\abs{A}^2\abs{T}^2 +$ $ \abs{A}\abs{T}^3)$.
\end{theorem}
\begin{proof}
    The same argument as Theorem~\ref{thm:min-var-square-relation} applies here, and
    \PrbSBAShort is equivalent to asking for an assignment minimizing the square sum of
    $\abs{\phi^{-1}(a)}$ for each $a \in A$.

    We transform this problem to the minimum-cost maximum flow problem with integral flows.
    Construct a multi-digraph $G=(V,E)$.
    The set of vertices $V$ includes $A$, $T$,
    the source $\alpha$ and sink $\beta$ of the flow.
    We construct the edges $E$ as follows.
    $A$ and $T$ preserve their relation $R$ as edges, oriented from $A$ to $T$.
    For every task $t \in T$, add one edge from $t$ to $\beta$.
    For every agent $a \in A$, add multiple edges from $\alpha$ to $a$ so
    that the number of incoming edges matches the number of outgoing edges at $a$.
    Only these edges from $\alpha$ have costs.
    For the $i^{\textnormal{th}}$ edge to each $a$, denoted $e_{a,i}$,
    we set the cost to $2i - 1$. All edges have capacity $1$.
    There are at most $2\abs{R} + \abs{T}$ edges in this graph,
    so the construction can be done in $\bigo{|A||T|}$ time.
    See Figure~\ref{fig:flow-conversion} for an example.

    Let $c(e)=1$ be the capacity and $w(e)$ be the cost of each edge.
    Also, let $\hat{w}$ be the minimum cost and
    $\hat{f}$ be the maximum flow amount of $(G,c,w)$ from $\alpha$ to $\beta$.
    We claim that for any $k \in \Z$, we get $\hat{w}\leq k$ and $\hat{f}=|T|$
    if and only if the \PrbSBAShort instance $(A,T,R)$ admits a valid assignment
    $\phi: T \to A$ with the square sum of the number of assigned tasks
    $\sum_{a \in A}\abs{\phi^{-1}(a)}^2 \leq k$.

    To prove the correctness in the forward direction, suppose that there
    exists a flow of amount $|T|$ and cost $\hat{w} \leq k$.
    The flow at $t \in T$ is $1$, and there is incoming flow from
    only one agent $a \in A$, which corresponds to assignment $t \mapsto a$.
    Likewise, the flow at $a \in A$ must equal $\abs{\phi^{-1}(a)}$. The minimum cost
    can be achieved by
    $\sum_{a\in A}\sum_{i=1}^{\abs{\phi^{-1}(a)}} (2i - 1) =
     \sum_{a \in A}\abs{\phi^{-1}(a)}^2 \leq k$.

    On the other hand, assume that there exists a valid
    assignment $\phi: T \to A$. Then, we can construct the following flow $f: E \to \Z$
    with amount $\abs{T}$ and cost $k$. Thus, $\hat{w} \leq k$, as desired.
    \vspace*{-0.5em}
    \begin{align*}
        f(e_{a,i}) &= \begin{cases}
            1 \quad \text{ if } i \leq \abs{\phi^{-1}(a)}\\
            0 \quad \text{ otherwise}
        \end{cases}\\
        f((a, t)) &= \begin{cases}
            1 \quad \text{ if } a = \phi(t)\\
            0 \quad \text{ otherwise}
        \end{cases} & \text {for every }a \in A\\
        f((t, \beta)) &= 1 & \text {for every }t \in T
    \end{align*}

    Since the capacities are integral, there exists an integral maximum flow \cite{ford1956maximal}.
    The Bellman-Ford algorithm can find an integer-valued minimum-cost maximum flow
    in time $\bigo{|f_{\max}||V||E|}$.
    In our case, that is $\mathcal{O}(|T|(|A|+|T|)|A||T|) = \mathcal{O}(|A|^2|T|^2+|A||T|^3)$.
\end{proof}

\noindent\begin{minipage}{\linewidth}
    \vspace*{1em}

    \definecolor{color1}{RGB}{99,110,250}
    \centering
    \makebox[\linewidth]{
        \begin{tikzpicture}

            \tikzstyle{lnode} = [circle, fill=white, draw, thick, scale=1, minimum size=0.2cm, inner sep=1.5pt]

            \tikzstyle{directed} = [color=black, arrows=- triangle 45]

            \node[lnode] (a1) at (0, 1.5) [label=left:$\alpha$] {};
            \node[lnode] (b1) at (6, 1.5) [label=right:$\beta$] {};
            \node[lnode] (x1) at (2, 0.5) {};
            \node[lnode] (x2) at (2, 1.5) {};
            \node[lnode] (x3) at (2, 2.5) {};
            \node[lnode] (y1) at (4, 0.5) {};
            \node[lnode] (y2) at (4, 1) {};
            \node[lnode] (y3) at (4, 2) {};
            \node[lnode] (y4) at (4, 2.5) {};

            \draw[directed] (a1) [label=$1$] to [out=-60,in=180] node[xshift=4pt, yshift=4pt, midway] {$1$} (x1);
            \draw[directed] (a1) [label=$1$] to [out=-80,in=210] node[xshift=-4pt, yshift=-4pt, midway] {$3$} (x1);

            \draw[directed] (a1) [label=$1$] to [out=80,in=160] node[xshift=-10pt, yshift=0pt, midway] {$1$} (x3);
            \draw[directed] (a1) [label=$1$] to [out=50,in=180] node[xshift=-8pt, yshift=0pt, midway] {$3$} (x3);
            \draw[directed] (a1) [label=$1$] to [out=25,in=210] node[xshift=-6pt, yshift=2pt, midway] {$5$} (x3);

            \draw[directed] (a1) [label=$1$] to [out=10,in=170] node[xshift=10pt,yshift=4pt,midway] {$1$} (x2);
            \draw[directed] (a1) [label=$1$] to [out=-10,in=190] node[xshift=10pt, yshift=-4pt, midway] {$3$} (x2);

            \draw[directed, color1, line width=1.2pt] (x1) -- (y1);
            \draw[directed, color1, line width=1.2pt] (x1) -- (y2);
            \draw[directed, color1, line width=1.2pt] (x2) -- (y2);
            \draw[directed, color1, line width=1.2pt] (x2) -- (y3);
            \draw[directed, color1, line width=1.2pt] (x3) -- (y2);
            \draw[directed, color1, line width=1.2pt] (x3) -- (y3);
            \draw[directed, color1, line width=1.2pt] (x3) -- (y4);
            \draw[directed] (y1) -- (b1);
            \draw[directed] (y2) -- (b1);
            \draw[directed] (y3) -- (b1);
            \draw[directed] (y4) -- (b1);

            \draw[dashed, rounded corners=1ex] (1.5, 0) rectangle (2.2, 2.8);
            \draw[dashed, rounded corners=1ex] (3.4, 0) rectangle (4.2, 2.8);
            \node[align=left, anchor=west] at (2.2, 2.8) {$A$};
            \node[align=left, anchor=west] at (4.2, 2.8) {$T$};
            \node[align=center, anchor=center] at (2.8, 0.2) {$R$};
        \end{tikzpicture}
    }
    \captionsetup{hypcap=false}
    
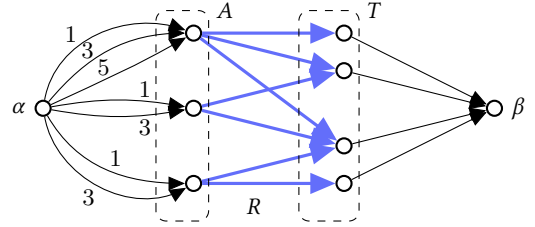
\captionof{figure}{%
    Transformation of \PrbSBAShort into a flow problem.
    The binary relation $R$ between agents $A$ and tasks $T$ is shown as blue edges.
    Every edge has unit capacity; non-zero edge costs are shown as numbers.}
    \label{fig:flow-conversion}
    \vspace*{1em}
\end{minipage}


We now prove Theorem~\ref{thm:runtime-bnp}, which states that BNP is tractable when
 the landmarks form a dominating set of $G$.

\begin{proof}[Proof of Theorem~\ref{thm:runtime-bnp}]
    Since non-landmarks can be assigned to any neighboring landmark, the problem is
    equivalent to \PrbSBAShort with agents $L$, tasks $V$, and relation
    $\{(u,u): u \in L\} \cup \{(u,v): u \in L, v \in N(u) \setminus L \}$.
    By Theorem~\ref{thm:runtime-bap}, the running time is\\ 
    $\bigo{|L|^2 |V|^2 + |L| |V|^3} = \bigo{|L|n^3}$.
\end{proof}



\subsection{\AlgRatio Approximation Guarantee}

In the interest of space, we sketch here only the main ideas of the proof of
Theorem~\ref{thm:approx-ratio-greedy}, which guarantees an $\mathcal{O}(\sqrt{\Delta^r})$-approximation.

\begin{proof}[Proof Sketch of Theorem~\ref{thm:approx-ratio-greedy}]
    First, observe that \AlgRatio behaves in the same way on $(G^r, 1)$ as on
    $(G, r)$. Since $\Delta(G^r) \leq \left(\Delta(G)\right)^r$, it
    suffices to consider $r=1$.

    Consider a partial dominating set $D \subseteq V$ and chosen vertex $v \in V$
    during an iteration of \AlgRatio. Let us call $t=\abs{N[D]}/\abs{V}$ the global
    domination ratio, and $\rho(v)=\abs{N[v] \cap N[D]}/\abs{N[v]}$ the local domination ratio.
    By definition, $v$ has the minimum local domination ratio among all vertices at this
    point.

    A key idea is to consider the minimum possible $t$ for the given $\rho$.
    Special cases lie in $0 \leq t \leq \frac{2}{\Delta+2}$ and $\frac{\Delta+1}{\Delta+2} \leq t \leq 1$,
    but they have constant effects on the approximation ratio, so we omit them here.

    Now, we may assume $0 < \rho < 1$. Let us partition $V$ into three parts: $D$, $U:=N[D] \setminus D$ and $W:=V\setminus N[D]$.
    Then, every vertex in $U$ may have at most $(1 -\rho)(\Delta + 1)$ neighbors in $W$,
    and every vertex in $W$ must have at least $\rho / (1 - \rho)$ neighbors in $U$,
    which leads to the relation: $|W|/|U| \leq (\Delta+1)(1-\rho)^2/\rho$.
    Consequently, we get: $t \geq f(\rho)$ and
    $\rho \leq 1 - f^{-1}(t)$, where $f(x) = 1/(1+\Delta(1-x)^2/x)$.

    Next, we relate $\avgcong(D)$ to $t$ by the fact
     that every time we add a vertex $v$ to $D$,
     $\avgcong(D)$ increases by $(\deg(v)+1)/n$ and
     $t$ increases by $(1-\rho(v))(\deg(v)+1)/n$.
    Finally, we obtain: $\avgcong_G \leq \bigo{1} + \int_{2/(\Delta+2)}^{(\Delta+1)/(\Delta+2)} \left(1/f^{-1}(t)\right)dt$,
    which evaluates to $\frac{\pi}{2}\sqrt{\Delta} + \bigo{1}$.
    Since $1 \leq \mac(G) \leq \avgcong(D)$ for all $G$, the approximation ratio is bounded by
    $\bigo{\sqrt{\Delta}}$ when $r=1$, and $\bigo{\sqrt{\Delta^r}}$ for general $r$.
\end{proof}

\subsection{Algorithm Details}
\label{sec:alg-details}

First, we give the running time of the algorithms from Section~\ref{sec:sparse-domset}.

\begin{theorem}
    \label{thm:dom-alg}
    \AlgDegree, \AlgDegreePlus, \AlgRatio, and \AlgRatioPlus find an $r$-dominating
    set in time $\bigo{\Delta^{2r}n \log n}$.
\end{theorem}

\begin{proof}[Proof of Theorem~\ref{thm:dom-alg}]
    All of these algorithms continue to \linebreak choose vertices until the entire
     graph becomes dominated, so the resulting set is an $r$-dominating set.

    The running time is also the same among these algorithms.
    Consider a max heap of $n$ elements, which
    stores values of the greedy criteria for each $v \in V$. We can find and remove the maximum
    element in $\bigo{\log n}$ time. After adding a vertex $v$ to a partial dominating set
    $D$, we need to update at most $N^{2r}[v]$ elements in the heap because
    all criteria are based on the number of undominated vertices in the $r$-neighborhood.
    There are at most $n$ iterations, and since $\abs{N^{2r}[v]} \leq \Delta^{2r}$,
     the total running time is $\bigo{\Delta^{2r}n \log n}$.
\end{proof}

Next, we show that the (compact) neighborhood kernel can be found in linear time.

\begin{algorithm}[!h]
    \DontPrintSemicolon
    \KwIn{Graph $G=(V,E)$, landmarks $L\in V$}
    \KwOut{Neighborhood kernel $H$}

    $E' \gets \emptyset$

    \For {$e \in E$}{
        $vw \gets e$ such that $d(v,L) \leq d(w,L)$

        \If {$d(v,L) + 1 = d(w,L)$}{\label{alg:line:nbr-kernel-cond}
            $E' \gets E' \cup \{vw\}$
        }
    }
    \KwRet{$H=(V,E')$}

    \caption[]{\hspace*{-4.3pt}{.} \AlgNbrKnl ($G, L$)}
    \label{alg:nbr-kernel}
\end{algorithm}

\begin{algorithm}[!h]
    \DontPrintSemicolon
    \KwIn{Graph $G=(V,E)$, landmarks $L\in V$}
    \KwOut{Compact neighborhood kernel $(H_c, \phi)$}
    $H \gets $ \AlgNbrKnl$(G, L)$\\
    Initialize $\psi: V \rightarrow V$ to the identity map

    \For {$v \in V$ sorted by ascending $d_G(v,L)$}{
        \If {$\deg_H^-(v) = 1$}{
            \label{alg:line:compact-nbr-kernel-cond}
            \textbf{let} $w \in N_H^-(v)$ \tcp*{only one in-neighbor}
            \textbf{set} $\psi(v) = w$ \tcp*{$v$ belongs to bag $w$}
            $H \gets H / vw$ \tcp*{contract edge $vw$}
        }
    }

    \KwRet{$(H, \psi^{-1})$}

    \caption[]{\hspace*{-4.3pt}{.} \AlgCmpNbrKnl ($G, L$)}
    \label{alg:compact-nbr-kernel}
\end{algorithm}

\begin{theorem}
    \AlgNbrKnl (\AlgCmpNbrKnl) correctly \linebreak gives the (compact) neighborhood kernel of the graph
    on the given landmarks in $\bigo{n+m}$ time.
\end{theorem}

\begin{proof}
   We begin with \AlgNbrKnl, where the check in Line~\ref{alg:line:nbr-kernel-cond} enforces
     Definition~\ref{def:nbr-knl}, ensuring correctness.
     For every vertex $v \in V$, the distance from the closest landmark
     can be computed by BFS from $L$ in $\bigo{n+m}$ time. The algorithm checks all edges in
     $\bigo{m}$, resulting in  $\bigo{n+2m}=\bigo{n+m}$.

    Turning to \AlgCmpNbrKnl, we traverse all vertices in BFS order from the
    landmarks, which guarantees each bag's representative to be the closest to $L$. By the
    definition of \PrbBNPShort, Line~\ref{alg:line:compact-nbr-kernel-cond} ensures that all
    bag members must be assigned to the same landmark because a vertex must follow its
    in-neighbor's assignment.
    Lastly, $\phi(v)$ is maximal for every $v \in V_c$ because otherwise,
     there exists a vertex $x \in V_c$ with at least two in-neighbors $y,z$ such that
     $v = \psi(y)=\psi(z) \neq \psi (x)$. This cannot happen because the algorithm visits
     $y,z$ before $x$ and contracts $y,z$ into one vertex in $H_c$.
    This modified BFS does not add any extra asymptotic running time
    to \AlgNbrKnl. Thus, the total running time is $\bigo{n + m}$.
\end{proof}

The following lemma is used to show that \AlgPrtBranch is FPT (Theorem~\ref{thm:bnp-fpt}).

\begin{lemma}\label{lem:compact-nbr-knl-in-degree}
    Given a graph $G=(V,E)$ and landmarks $L \subseteq V$ with neighborhood kernel
    $H=(V,E')$ and compact neighborhood kernel $(H_c=(V_c,E_c), \phi)$,
    $\deg_{H_c}^-(v) \leq \deg_H^-(v) \leq \Delta(G)$ for every $v \in V_c$.
\end{lemma}

\begin{proof}
    It is clear to see $\deg_H^-(v) \leq \deg_G(v) \leq \Delta(G)$.
    Also, \AlgCmpNbrKnl performs only contractions on $H$, so\linebreak
    $\deg_{H_c}^-(v) \leq \deg_H^-(v)$.
\end{proof}

\begin{algorithm}[!h]
    \DontPrintSemicolon
    \KwIn{Graph $G=(V,E)$, landmarks $L=\{u_1,\ldots,u_{\ell}\} \subseteq V$}
    \KwOut{Neighborhood partitioning $\mathcal{A}$}

    $H \gets $ \AlgNbrKnl$(G)$\\
    Partition $V$ into layers $\{V_i=\{v \in V: d(v,L)=i\}\}$\\
    Initialize map $f: V \to L$ to $f(u)=u$ for all $u \in L$\\

    \For {$i \gets 1$ \KwTo $\max_{v \in V} d(v,L)$} {
        \tcp{initialize relations for assigned vertices}
        $R \gets \{(f(x),x) \mid x \in \bigcup_{j=0}^{i-1} V_j\}$\\
        $R \gets R \cup \{(f(w), v) \mid v \in V_i, w \in N_H^-(v)\}$

        update $f$ with SBAP($L$, $\bigcup_{j=0}^i V_j$, $R$) \label{alg:line:prt-layer-flow}
    }

    \KwRet{$f^{-1}$ \tcp* {map from $L$ to assigned vertices}}

    \caption[]{\hspace*{-4.3pt}{.} \AlgPrtLayer ($G, L$)}
    \label{alg:prt-layer}
\end{algorithm}

\section{Experimental Setup}
\label{sec:exp-setup}


Datasets, along with a summary of their invariants, are available at~\cite{codebase}.
Metagenomic cDBGs were generated by BCALM 2 (v2.2.1)\cite{10.1093/bioinformatics/btw279}
with k-size 31. Some non-metagenome networks originated from the Network Data Repository \cite{nr}.
For disconnected graphs, we used the largest connected component.


We ran all experiments on identical hardware, equipped with 40 CPUs
(Intel(R) Xeon(R) Gold 6230 CPU @ 2.10GHz) and 190 GB of memory,
and running CentOS Linux release 7.9.2009. \sgc is written in Python 3.
We used Gurobi Optimizer 9.1.2 for ILP and QP. We used Plotly for creating charts.
Our algorithms are written in C++ with OpenMP and compiled with gcc 10.2.0.


Experimental results are fully replicable using the code and data
at~\cite{codebase}; detailed instructions in \texttt{README.md}.

\end{document}